\documentclass{article}

\usepackage{arxiv}

\usepackage[utf8]{inputenc} 
\usepackage[T1]{fontenc}    
\usepackage{hyperref}       
\usepackage{url}            
\usepackage{booktabs}       
\usepackage{amsfonts}       
\usepackage{nicefrac}       
\usepackage{microtype}      
\usepackage{lipsum}
\usepackage{graphicx}

\usepackage{subcaption}
\usepackage{amsmath}
\usepackage{amsthm}
\graphicspath{ {./images/} }

\newtheorem{theorem}{Theorem}[section]
\newtheorem{lemma}[theorem]{Lemma}
\newtheorem{corollary}[theorem]{Corollary}

\title{Modeling Tradeoffs between mobility, cost, and performance
in Edge Computing}

\author{
 Muhammad Danish Waseem and Ahmed Ali-Eldin \\
  Chalmers university of Technology\\
  \texttt{\{danishm, ahmed.hassan\}@chalmers.se} \\
}

\begin{document}
\maketitle
\begin{abstract}
  Edge computing provides a cloud-like architecture where small-scale resources are distributed near the network edge, enabling applications on resource-constrained devices to offload latency-critical computations to these resources. While some recent work showed that the resource constraints of the edge could result in higher end-to-end latency under medium to high utilization due to higher queuing delays, to the best of our knowledge, there has not been any work on modeling the trade-offs of deploying on edge versus cloud infrastructures in the presence of mobility. Understanding the costs and trade-offs of this architecture is important for network designers, as the architecture is now adopted to be part of 5G and beyond networks in the form of the Multi-access Edge Computing (MEC).
In this paper we focus on quantifying and estimating the cost of edge computing. Using closed-form queuing models, we explore the cost-performance trade-offs in the presence of different systems dynamics. We model how workload mobility and workload variations influence these trade-offs, and validate our results with realistic experiments and simulations.  Finally, we discuss the practical implications for designing edge systems and developing algorithms for efficient resource and workload management.
\end{abstract}


\section{Introduction}

Edge computing offers an alternative to traditional cloud offloading by deploying compute and storage resources at the network edge providing superior network latencies compared to cloud systems for offloading. The promise is that greater proximity to end-devices should address the latency-sensitive needs of real-time and safety critical applications.
One popular edge model, involves deploying small server clusters or micro-data centers at geographically \emph{distributed locations} providing a cloud-like service to applications. Such a distributed edge cloud model can provide many advantages over traditional cloud computing. 
First, with multiple geo-distributed sites,  users at different locations choose the nearest edge sites, enabling the load to be geographically distributed across the edge sites while providing low-access latency. Second, as mobile users move locations, their devices can choose a new edge location that is more proximal than the previous one, enabling edge clouds to handle mobile workload dynamics.  Finally, localized computations on the edge can provide increased privacy~\cite{garcia2015edge}.

However, these advantages of distributed edge clouds come with a number of trade-offs; Since edge clouds are comprised of small clusters that are dispersed across a large number of geographic locations, they have lower economies of scale compared to traditional clouds that co-locate these servers at very large data centers requiring, e.g., increased construction, installation, maintenance, and management costs. 
Second, since each edge site comprises of a small cluster, the potential for overloads at a particular site, and geographic imbalances across various edge sites is non-trivial due to the skewed nature of geo-distributed workloads. This results in poor application performance or higher capital costs to provision additional hardware resources. In contrast, the large pool of servers in traditional cloud data centers offers better multiplexing of resources across skewed workloads.  Specifically, statistical multiplexing benefits in traditional clouds arise from workloads with different spatial (different areas, cities, countries, etc), and temporal (when workload bursts occur in one workload) characteristics being co-located on the same infrastructure, allowing for resource sharing. In edge clouds, such workload dynamics cause overloads or geographic imbalances. A consequence of these two drawbacks is that distributed edge clouds incur higher costs and are more expensive to build than traditional cloud platforms. It is not apriori evident whether the performance benefits offered by edge clouds  outweigh these higher economic costs. 

One previous study~\cite{ali2021hidden} has explored the edge-cloud "performance inversion problem".  Using both analytical modeling and experimental results, the authors highlight that while edge computing offers lower network latency due to its proximity to end-users, its resource limitations can lead to higher end-to-end latency compared to centralized cloud data centers, especially under higher utilization. The authors have referred to the performance inversion problem as the hidden cost of edge computing. However, this prior work does not account for mobility overheads in edge networks, a crucial factor in dynamic environments where devices frequently move or change networks. Mobility introduces additional challenges, such as handovers and context migrations, which can further degrade performance. These migration overheads can be significant with, e.g.,  up to 14.5x increase in latencies for streaming services during handovers as reported for live 5G network deployments~\cite{hassan2022vivisecting}. 

Another missing aspect in the cloud versus edge discussion is a characterization of the additional capacity required for edge infrastructures to maintain (better) performance relative to a centralized cloud due to both mobility and workload skews. Understanding these capacity requirements is crucial for effectively scaling edge computing resources to handle increased workloads and ensuring that performance benefits are sustained as the demand for edge services grows.

In this paper, we take a step towards modeling the performance and costs of realistic edge cloud deployments with migration. Our focus is on finding closed-form models to better study the real performance and monetary costs of edge computing infrastructures in realistic scenarios. Our aim is to provide researchers, system designers, and network operators with models that can be used to analyze different edge design choices. Our contributions can be summarized as follows:
\begin{itemize}
    \item We develop closed-form analytical models to study and quantify the various cost and performance trade-offs with mobility overheads for edge versus cloud deployments.
    \item We quantify the additional capacity requirements for edge infrastructure needed to match or exceed the performance of centralized clouds, considering factors such as workload distribution and mobility.
    \item We validate the accuracy of our models conducting experiments and simulations.
    \item We discuss the practical implications of our results, providing insights into how our findings can influence the design of edge-cloud systems.
\end{itemize}

\section{Background}
\subsection{Edge Computing Systems}
The idea of using edge resources for their lower latency can be traced to the late nineties, with work on Content-Distribution-Networks (CDNs) and caching at the edge of the network~\cite{karger1997consistent,nygren2010akamai,dilley2002globally}. However, for offloading computations, edge computing can be viewed as an extension to the cloudlets concept~\cite{satyanarayanan2009case}. In the cloudlet vision, a mobile user instantiates service software using virtualization technologies on a wireless LAN accessible \emph{cloudlet} nearby. 
A cloudlet is a resource rich computing cluster at the network edge. 
The idea of edge computing further evolved with significant interest from both academia and industry. A key milestone was the adoption of the \emph{Multi-access edge computing (MEC) standard} by ETSI and the 3GPP as a standard architecture and a building block of 5G~\cite{etsi5g,3gpp}. 
This has fueled research in the area~\cite{harchol2020making}, with some focus on video streaming and analytics applications~\cite{ni2023cellfusion}, visual SLAM~\cite{xu2022swarmmap}, and autonomous vehicles~\cite{baidya2020vehicular}. 
In addition, recently, more and more cloud providers have started offering different versions of the edge computing idea~\cite{awsedge,gcpedge,azureedge} 

While Cloud (and Telco) operators have rushed to sell edge computing resources, the promise of edge clouds is far from being realized. 
We believe that this is partially due to the lack of accurate models that capture the cost-performance-latency trade-offs of different edge design choices. Such models will have implications on both how both the system  and edge applications should be designed. 
While there are different implementations of the edge cloud architecture varying from using 5G to Wifi-6~\cite{zhang2024measurements} and LoRa~\cite{rubambiza2023comosum},  in essence, these systems can be modeled using queuing theory~\cite{ali2021hidden}. 
It has been shown both using modeling and experiments that, while the edge does provide superior performance to the cloud at lower utilization, once utilization increases, queuing delays increases at the edge quickly starts dominating the overall service latency of applications~\cite{ali2021hidden}. At medium to high utilizations, the edge can have orders of magnitude higher latency compared to sending the requests to the cloud, leading to "Performance Inversion".

\subsection{Edge Challenges} 
While performance inversion is a challenge, careful and intelligent scheduling can improve the performance of edge deployments such that performance inversion is minimized or eliminated~\cite{wang2024invar}. However, we believe that the limitations of edge cloud architectures are more than just performance inversion, and that rushing to deploy these systems without careful analysis of their performance is premature. In this paper, we focus on three main challenges we have identified with edge computing systems.

First, one of the main selling points to use cloud computing  is the ease of multiplexing workloads with varying temporal 
and spatial dynamics~\cite{khan2012workload,lu2017imbalance}. Temporal workload dynamics are usually a result of an underlying pattern of usage, e.g., day-and-night effects seen in many workloads where there is an increased load during certain hours of the day~\cite{joosen2023does,benevenuto2009characterizing}. Some of these workload dynamics can occur as short term burstiness and spikes in the load due to sudden flash crowds~\cite{atikoglu2012workload,wendell2011going}. 
While spatial dynamics can occur due to the differences in, e.g., time-zones; for edge clouds, we are more interested in spatial dynamics that arise within the same geographical region due to user mobility patterns~\cite{hassan2022vivisecting}. 

The second challenge we have identified arises from this mobility. In the case of 5G networks, when a user moves between two base-stations, a handover to move the data and conncetion of the user is required. Similarly, in the case of edge networks, when a user moves between two edge, the user computations (and state if available) need to migrate between the two edges. Recently, it has been shown that today's 5G networks fail to manage these handovers without significant loss in the Quality-of-Service for video-streaming and online-gaming~\cite{hassan2022vivisecting}. If each of these base-stations hosts also an edge site, a design that has been pushed by most major telcom companies, these handovers would be further complicated by adding computations on edge resources as these computations would also be needed to migrate, not just the communication data which is usually stateless or has very little state. 

The final challenge we aim to address is to quantify how much does edge systems cost compared to traditional cloud systems. Edge cloud systems will be more expensive for operators compared to running on traditional large-scale clouds due to the increased costs of maintenance, installations, and security. Unlike a datacenter with controlled access and physical security, the distribution of server resources does not provide similar security. More importantly, the results from Ali-Eldin et al~\cite{ali2021hidden} show that these servers need to run at a lower utilization compared to the utilization at the cloud. This extra cost due to running at lower utilization has never been studied in details.

\subsection{The Real Cost of the Edge}

Our aim is to build models to study the three challenges stated above focusing on the tradeoffs of edge versus cloud computing. 
We start with the simple queuing model of edge versus centralized clouds used by Ali-Eldin et al.~\cite{ali2021hidden} shown in
Figure \ref{fig:qs}. In this model, the aggregate workload, $\lambda$, arriving at a centralized cloud data center gets equally partitioned across $k$ edge cloud locations, i.e., they assume no spatial or temporal skews, and no migrations.
For edge resources, we model all the servers
resident at an edge cloud location as a single logical server. Since the centralized cloud aggregates all of these servers into a single data center,  the centralized cloud is modeled as $k$ logical servers, which represents the collection of servers from all $k$ edge cloud sites. Previous queuing models show that it is usually beneficial to combine multiple queues to reduce waiting times,  with  a few known exceptions, e.g., when jockeying between the queues is permitted~\cite{rothkopf1987perspectives}. 

For example, it is a known result that, at high system utilization close to 100\%, an M/M/k system with a single queue yields waiting times that are factor of $k$ smaller than using $k$ M/M/1 servers each with a single queue~\cite{harchol2013performance}.
Thus, we expect the queuing delays to be lower in the centralized cloud than in the distributed edge cloud for the same workload. However, each edge cloud data center also has lower network latency to the end user compared to a remote centralized cloud data center. Assuming that the service times are identical in both the edge cloud servers and centralized cloud servers (e.g.,  both use the same server hardware), the benefits of a distributed edge cloud are realized only when the higher queuing delays due to distributed queues are more than offset by the lower network delays offered by each edge site. 

\noindent\textbf{The Cost of Migrations.} In edge computing scenarios, when a user moves from one edge location to another, there is an additional mobility overhead that affects the response time $T_{edge}$. For example, consider an autonomous vehicle equipped with computer vision systems that rely on edge computing for some of its real-time processing and for Vehicle-to-Infrastructure~\cite{liu2021dart}. As the vehicle moves from one city block to another, it must transfer its current processing tasks and data-such as the analysis of surrounding traffic conditions and the ongoing route planning—from the edge server in the previous block to the edge server in the current block. These handovers can introduce delays due to the time needed to reprocess or transfer this information, thereby increasing the response time. 
This leads to the \textbf{first research question (RQ1)} addressed in this paper: \emph{What is the increased cost of data migration between edge sites?}
\begin{figure*}
   \begin{minipage}{0.49\textwidth}
   \begin{subfigure}{0.4\columnwidth}
        \centering
        \includegraphics[width=\linewidth]{ 
     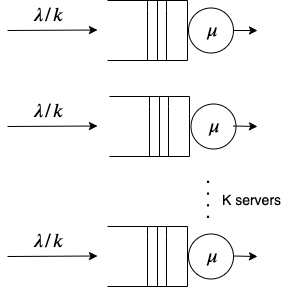}
        \caption{Edge queuing model.}
        \label{fig:distQ}
    \end{subfigure}
     \begin{subfigure}{0.5\columnwidth}
        \centering
        \includegraphics[width=\linewidth]{ 
     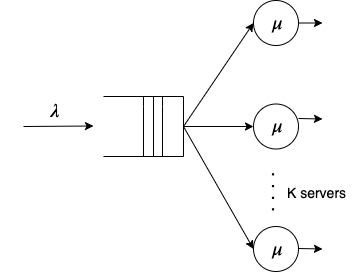}
        \caption{Centralized large-scale cloud queuing model.}
        \label{fig:centQ}
    \end{subfigure}    
    \caption{Queuing models for distributed edge clouds versus centralized remote large-scale backend clouds.}
    \label{fig:qs}
\end{minipage}
   \begin{minipage}{0.49\textwidth}
        \centering
        \includegraphics[width=\linewidth]{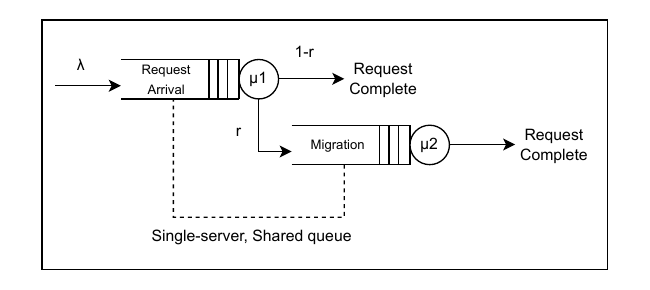}
  \caption{Two-phased single-server queue. Requests arrive and receive the first service (\(\mu_1\)). Then, with probability \( r \), they enter the migration phase (\(\mu_2\)); otherwise, with probability \( 1 - r \), they exit the queue.}
    \label{fig:2phase}
        \end{minipage}
    \end{figure*}

\noindent\textbf{Impact of Workload Skews.}   
The simplified model assumes that the workload arriving at a centralized cloud is equally partitioned across $k$ edge cloud sites and that the load is balanced across all locations. However, as noted earlier, cloud workloads exhibit spatial and temporal dynamics that lead to workload skews and utilization differences across edge locations due to the mobility of users~\cite{gonzalez2008understanding}. 
Skewed workload causes edge sites to experience differing utilization levels. Sites with greater incoming load will experience higher queuing delays 
due to higher utilization. In contrast, in centralized clouds, all requests arrive into a single data center that houses all servers, enabling the system to multiplex bursty workloads onto a larger number of servers, and hence better absorb the impact of spatial and temporal workload dynamics.
This leads to our \textbf{second research question (RQ2)}: \emph{how do spatial and temporal workload skews impact the queuing delays seen within a distributed edge cloud and how do these compare to that in centralized clouds? }

\noindent\textbf{The Cost of Resource Multiplexing and Co-location.} 
Our final observation is based on the observations made by Ali-Eldin et al.~\cite{ali2021hidden} using the well-known \textit{empirical rule} (also known as the \textit{68-95-99.7 rule}) from statistics. Consider a distributed edge cloud with $k$ sites and assume that the request arrival rate at each site follows an independent Poisson process with a mean arrival rate of $\lambda$. 
The Poisson distribution converges to a normal distribution as $\lambda$ increases. Consequently,  95\% of the values fall within two standard deviations of the mean, i.e., the 95th percentile of a normal distribution is approximately $\mu + 2\sigma$, where $\mu$ is the mean and $\sigma$ is the standard deviation.
For a Poisson distribution, the standard deviation is the square root of the mean, i.e., $\sigma = \sqrt{\lambda}$. Consequently, for large values of $\lambda$, the 95th percentile of the arrival rate can be approximated as: $ C_{edge}=k(\lambda+2\sqrt{\lambda}).$

In case of the centralized cloud, all requests arrive into a single data center and the total workload is the summation of the workload seen at each edge site.  Since the arrival rates are independent Poisson random variables, and since the summation of $i$ independent Poisson random variables is a Poisson random variable, with a mean arrival rate $\sum_i \lambda_i$, it follows that the workload seen by the centralized cloud is a Poisson
process with a mean request rate $k\lambda$ and a standard deviation of $\sqrt{k\lambda}$.
 Applying the two sigma rule, the 95-th percentile of the arrival rate  
 is $k\lambda+2\sqrt{k\lambda}$, for which the provisioned capacity $C_{cloud}$ to handle this peak workload should be
$C_{cloud}=k\lambda+2\sqrt{k\lambda}.$

Clearly, $C_{cloud} < C_{edge}$ since $\sqrt{k\lambda} < k\sqrt{\lambda}$. 
This analysis  shows the {\em smoothing benefits} in a centralized cloud since combining workloads from multiple edge cloud sites yields a smoothed workload with a lower peak than the sum of the individual peaks. Consequently, the cost of deploying a distributed edge cloud is higher than that of the centralized cloud (without even including the additional benefits of economies of scale incurred in building larger data centers). 

However, the above analysis is based on the M/M/k model and does not include the fact that the edge resources might be shared between multiple applications. Understanding how such resource sharing can affect the costs is important as it can lead to a completely different design of the network.
 These observations lead us  our \textbf{third and final research question (RQ3)}: \emph{how much extra capacity needs to be provisioned in a distributed edge cloud to service the same workload seen by a centralized cloud, and how does these capacity requirements change in the presence of skewed demand?}

\section{Edge Mobility Overheads}
\label{sec:latency-model}
In this section, we aim to answer RQ1. We start by considering a distributed edge cloud with $k$ geographic locations. Let $\lambda$ denote the mean request rate seen by each edge cloud location. For now, we assume that the workload is balanced across all edge cloud data centers. 
A centralized cloud servicing the same workload as the distributed edge cloud sees a request rate of $k\lambda$.  Let $t_{edge}$ and $t_{cloud}$ denote the network round trip time (RTT) to the edge and the cloud data center, respectively. Since edge data centers are more proximal to end-users,  $t_{edge} < t_{cloud}$.  Let $w_{edge}$ and $w_{cloud}$ denote the queuing delay (i.e., waiting time) incurred by a request upon arrival at the edge and the cloud data center, respectively. Finally, let $s_{edge}$ and $s_{cloud}$ denote the service time of a request at the edge and cloud data center, respectively. The response time of a request is then the summation of the network RTT, queuing delay, and service time. Therefore, for edge to have better response time than cloud:
\begin{equation}
  t_{edge}+w_{edge}+s_{edge} <  t_{cloud}+w_{cloud}+s_{cloud}
  \label{eq:lessthan}
\end{equation}





\subsection{Mobility Overhead} \label{sec:time-skew}

As previously stated, it is expected that a significant number of edge applications will involve user mobility. 
Similar to mobile phone handovers, there will be handover and migration delays due to downtime during the handover \cite{ha2017you, chaufournier2017fast}.
Handovers can involve varying levels of complexity. For instance, in a stateless migration, such as relocating a web service from one edge node to another without retaining any user-specific data or session information, the overhead is minimal. 
In contrast, a stateful migration, such as handling real-time video streaming data, can introduce significant overhead.
During a handover, resources such as computational power, network bandwidth, and memory are consumed, which were previously available for handling incoming requests. This migration overhead can lead to increased response times and reduced throughput not only for the migrating user's requests but also for other requests in the system. 
 


To model the mobility overhead on a single edge site, we model each edge site as two-phased  single-server queues as shown in Figure~\ref{fig:2phase}. In this model, request serviced can occur with either no migrations, passing through the first phase only of the two-phased queue, or with migrations, having to rely on the migration logic modeled as the second phase in the queuing model which captures, e.g., the migration algorithm, the network transfer bandwidth for migrations, and the size of the data to be migrated. This model is general and can capture any overheads incurred at the edge for migration. An arriving request at the two-phased queuing system can thus either be a new request from a \emph{home-user} or a request from a \emph{migrating user}, which adds to the variability in the request arrival rate \emph{$\lambda$} as we discuss later. 

Let $s_{migration}$ be the service time of handling migrations on an edge site, i.e., the second phase on the edge described above. Request migrations are not common in centralized clouds unless there are failures, a case that we do not aim to cover in our models. Equation~\ref{eq:lessthan}  becomes:
\begin{equation}
t_{edge}+w_{edge}+s_{edge}+s_{migration}  <  t_{cloud}+w_{cloud}+s_{cloud}
  \label{eq:deltan}
\end{equation}
Assuming that both the edges and centralized cloud have the same hardware configuration with First-Come-First-Serve service discipline, it follows that the service time for a request
will be identical on both the edge and the cloud server once a request starts being processed. That is, $s_{edge} = s_{cloud}$, with only the mobility overhead as a difference between the two. The inequality \eqref{eq:deltan} 
 then reduces to
\begin{equation}
    t_{edge}+w_{edge}+s_{migration}  < t_{cloud}+w_{cloud}   
\end{equation}
Let $\Delta t$ denote the difference in network round trip times between the edge and cloud data centers. That is 
$\Delta t =  t_{cloud} - t_{edge}$. Substituting $\Delta t$ in the above inequality yields:
\begin{equation}
   \Delta t >  w_{edge} - w_{cloud} + s_{migration}\label{eq:deltan2}
\end{equation}

\subsection{M/M/k Latency bound}  \label{MMk Latency bound}

We first assume that each edge data center is modeled as single  $M/M/1$ server. 
Since there are $k$ edge sites, the centralized cloud is then modeled as a $M/M/k$ system.
\begin{lemma}
Assuming that the edge and cloud data centers are modeled as $M/M/1$ and $M/M/k$ queuing systems, respectively, The edge cloud offers lower response times whenever the network RTT difference between the edge and the cloud is greater than $(\frac{\lambda(\mu_2^2 +r\mu_1^2+r\mu_1\mu_2)}{\mu_1\mu_2(\mu_
1\mu_2-\lambda\mu_2-r\lambda\mu_1)} + \frac{r\lambda}{\mu_2 (\mu_2- r\lambda)} + \frac{r}{\mu_{2}} - \frac{1}{\mu_{cloud}(1-\rho_{cloud})\sqrt{k}}) $
\end{lemma}
\textbf{Proof: }
To incorporate mobility into the edge cloud model, we use the analysis of a two-phase queueing system with an optional second service phase, as described in \cite{al2003m}, which is an extension of the well-known Pollaczek–Khinchine formula \cite{bolch2006queueing}. Client sends a service request to the closest edge data center, where it undergoes the first phase of service. After this initial service is complete, with probability $r$, the client may migrate from the current edge site to another, requiring a second phase of service. Alternatively, with probability $1-r$, the client may exit the system, remaining in the same edge location during the entire processing (Figure \ref{fig:2phase}). 

Waiting time for exponential service and migration time distributions (i.e., \( M/M/1 \) queue) is given by \cite{al2003m}:
\begin{equation}
   \mathbb{E}[w]_{(M/M/1)} =   \frac{\lambda \left( \frac{2}{\mu_1^2} + \frac{2r}{\mu_2^2} + \frac{2r}{\mu_1 \mu_2} \right)}{ 2\left( 1 - \frac{\lambda}{\mu_1} - \frac{r \lambda}{\mu_2} \right)} \quad \Rightarrow \quad 
\frac{\lambda(\mu_2^2 +r\mu_1^2+r\mu_1\mu_2)}{\mu_1\mu_2(\mu_
1\mu_2-\lambda\mu_2-r\lambda\mu_1)}\label{eq:deltan5}
\end{equation}

For exponential service and migration times at the edge, the service time for requests that do not migrate is, \( s_{\text{edge}} = \frac{1}{\mu_{1}} \), where \( \mu_{1} \) is the service rate for non-migrating requests.  
If there is no migration, i.e.\ $r = 0$ or migration does not take any time i.e. $\mu_2 \rightarrow \infty$, then the non-simplified form of Equation~\eqref{eq:deltan5} becomes
\( \frac{\lambda}{\mu_1(\mu_1 - \lambda)} \label{eq:deltan10}
\) 
which is a known result of the M/M/1 queue \cite{bolch2006queueing}.

Additionally, the utilization of the destination edge server also significantly impacts the waiting time for those requests that are migrated. Since the destination site is also an $M/M/1$ queue, its expected waiting time is also given by equation\eqref{eq:deltan5}. To calculate the actual service time for requests that are migrated, we need to take into account the destination site queuing time. However, to ensure that only the contribution of migrated requests is accounted for in the waiting time, and assuming that a request migrates only once before being completely serviced, and since the request arrival and service rates across all edge sites are equal, we set $\mu_2 = 0$ for the destination queue, i.e., the request migrated will not be migrated again. Then,
\[
\mathbb{E}[w_{dest}]_{(M/M/1)} = \frac{r\lambda}{\mu_1 (\mu_1- r\lambda)}
\]

Thus, the total expected waiting time at the edge is given by:

\begin{equation}
   \mathbb{E}[w_{edge}]_{(M/M/1)} =  \frac{\lambda(\mu_2^2 +r\mu_1^2+r\mu_1\mu_2)}{\mu_1\mu_2(\mu_
1\mu_2-\lambda\mu_2-r\lambda\mu_1)} + \frac{r\lambda}{\mu_1 (\mu_1- r\lambda)}\label{eq:deltan6}
\end{equation}






Since the second phase of the queuing model does not only lead to waiting delays, but there are also processing delays, e.g., for the migration algorithms, this also needs to be factored in our model. Given that only a proportion \( r \) of requests continue to phase two and require to migrate to different edges, the expected service time for phase two (the migration), averaged across all requests (including those that skip phase two), is \( s_{\text{migration}} = r \times \frac{1}{\mu_{2}} \).

Finally, to analyze queuing delays at the centralized cloud in Equation~\ref{eq:deltan2} having $k$ servers, we use the model developed by  Whitt et al. ~\cite{whitt1992understanding} in the Quality-and-Efficiency-Driven (QED) regime \cite{halfin1981heavy}. The model characterizes the conditional expected waiting times in large-scale multiserver queues operating under heavy traffic as,
\begin{equation}
    \mathbb{E}[w_{cloud}|w_{cloud}>0]_{(M/M/k)}  =\frac{1}{\mu_{cloud}(1-\rho_{cloud})\sqrt{k}} \label{eq:ExpWM}
\end{equation}
Since \eqref{eq:ExpWM} represents the expected waiting time of requests conditioned on the waiting time being greater than zero, it follows that $
\mathbb{E}[w_{\text{cloud}} \mid w_{\text{cloud}} > 0] \geq \mathbb{E}[w_{\text{cloud}}]$,
because the conditional expectation excludes requests with zero wait time, thereby providing a conservative bound for $\Delta t$. 
Putting $w_{edge}$, $s_{\text{migration}}$ and $w_{cloud}$ in \ref{eq:deltan2}, we get,

\begin{equation}
   \Delta t >   \frac{\lambda(\mu_2^2 +r\mu_1^2+r\mu_1\mu_2)}{\mu_1\mu_2(\mu_
1\mu_2-\lambda\mu_2-r\lambda\mu_1)} + \frac{r\lambda}{\mu_1 (\mu_1- r\lambda)}  + \frac{r}{\mu_{2}} - \frac{1}{\mu_{cloud}(1-\rho_{cloud})\sqrt{k}},\label{eq:deltan6}
\end{equation}
which completes the proof.



\begin{corollary} \label{corollary43}
As the number of edge locations $k$ increases,  $\Delta n$ becomes a function of only the edge cloud utilization. 
\end{corollary}
\noindent\textit{Proof: }
As $k\rightarrow \infty$, the term $\frac{1}{\mu_{cloud}(1-\rho_{cloud})\sqrt{k}} \rightarrow 0$. Therefore, as the number of edge locations $k$ increases,  $\Delta t$ becomes a function of only the edge cloud workload parameters. 
\begin{equation}
   \Delta t >  w_{edge} + s_{migration}\label{eq:deltan4}
\end{equation}   
\begin{corollary} \label{corollary2}
Utilization of the destination server significantly affects the waiting time.
\end{corollary}
\noindent\textit{Proof: }
\[
\mathbb{E}[w_{dest}]_{(M/M/1)} = \frac{r\lambda}{\mu_1 (\mu_1- r\lambda)} \Rightarrow   \frac{r \rho_{dest}}{\mu_1 (1 - \rho_{dest})}
\]

As $\rho_{dest} \rightarrow 1,  \frac{r \rho_{dest}}{\mu_1 (1 - \rho_{dest})} \rightarrow \infty$. 
If the destination edge server lacks the capacity to handle further incoming requests, they should be offloaded to the cloud. Similarly, the best case scenario for edge would be when $\rho_{dest} \rightarrow 0$, then the term $\mathbb{E}[w_{dest}]_{(M/M/1)} \rightarrow 0$  

\begin{corollary} \label{corollary3}
Suppose that $t_{edge}$ is approximately equal to zero. The above inequality becomes the lower bound on the cloud network round-trip time below which the edge cloud will yield worse response time than the centralized cloud.
\end{corollary}
\noindent\textit{Proof:} 

if $t_{edge}=0$, $\Delta t$ becomes equal to the cloud latency only in Equation~\ref{eq:deltan2}.

It is evident that when the queuing delays on centralized cloud and latency of edge cloud are negligible, then this implies that if the cloud latency drops below a certain threshold, it has the potential to offer overall response times that are ``good enough'' for edge applications and lower than what a smaller edge data center can provide, since the reduction in network latency at an edge site no longer offsets the higher queuing delays in the edge data center. Queuing delay on edge data center is heavily effected by the frequency of migrations $r$ and migration service rate $\mu_2$.

\noindent\textbf{Implications:}
When deciding when to offload to the cloud, equation \eqref{eq:deltan6} becomes a simple threshold policy: offload when the expected delay at the edge exceeds the cloud delay + latency penalty.
In equation \eqref{eq:deltan5}, the numerator is the service variability. Every extra percentage of tasks that migrate adds variability  as $(r/\mu_2^2 + r/\mu_1\mu_2)$. Slower $\mu_2$ inflates both migration‑related terms, so variability can explode when the second phase is a bottleneck. This could be  problematic for latency-sensitive applications. For instance, for ML inference tasks,  task completion time varies depending on model size. Since these workloads often require large models (and thus higher startup latency), the inflation of variability becomes significant. One mitigation is to reduce $r$ by routing  requests through sticky sessions when the user is within a high-confidence geo-fence. However, again there would be trade-off of latency under client mobility.

The denominator in  equation \eqref{eq:deltan5} is the remaining capacity of the edge which must remain above $0$ for the system to remain stable. Even if the first phase is fast ( $\mu_1$ large), a moderate 
$r$ and/or slow second phase ( $\mu_2$  small) can push utilization $\rho \rightarrow 1$ much sooner. Our tuning levers are to reduce $r$ or increase $\mu_2$  to increase $1-\rho$, lowering delay. 
In edge computing deployments with mobile clients,  capacity planning and  handover management are crucial to ensure low-latency service during user migrations. Our two-phase queueing model yields actionable guidelines for edge designers by characterizing how migration probability $(r)$ and migration service rate $(\mu_2)$ affect response times.


We show that even with the simplest workload assuming Poisson arrivals, migration could significantly affect the performance of the edge cloud. 
\subsection{GI/G/k Latency Bound}
The  $M/M/k$ latency bound relies on idealized assumptions of Poisson arrivals and exponentially distributed service times and it provides valuable analytical insight due to its simplicity. However, the landmark works by Leland et al. \cite{leland2002self} and Paxson and Floyd \cite{paxson1995wide} demonstrated that real-world network traffic exhibits long-range dependence, and heavy-tailed behavior, all of which violate the assumptions underpinning Markovian models. To better reflect realistic conditions, we extend our latency analysis in section \ref{MMk Latency bound} to approximate a $GI/G/k$ system. 
We again assume each edge data center to be modeled as a single $GI/G/1$ queue and the cloud data center to be modeled as a $GI/G/k$ system that services the aggregate workload of the $k$ edge data centers.
As discussed by Whitt et al. in~\cite{whitt1993approximations}, although the $GI/G/k$ model assumes independent interarrival times and relies on a parametric characterization through a renewal process, the aim is not to overlook the dependence between successive arrivals. Rather, the essential impact of this dependence is captured indirectly by incorporating a variability parameter, such as the squared coefficient of variation.
\begin{lemma}
Assuming that the edge and cloud data centers are modeled as $GI/G/1$ and $GI/G/k$ queuing systems, respectively, The edge cloud offers lower response times whenever the network RTT difference between the edge and the cloud is greater than $( [
\frac{\lambda \left( \frac{1}{\mu_1^2} + \frac{r}{\mu_2^2} + \frac{r}{\mu_1 \mu_2} \right)}{1 - \frac{\lambda}{\mu_1} - \frac{r \lambda}{\mu_2}}
+ \frac{r\lambda}{\mu_1 (\mu_1- r\lambda)}] \cdot  
\frac{c_A^2 + c_S^2}{2}
 + \frac{r}{\mu_2} - 
\frac{\rho_{\text{cloud}}^k + \rho_{\text{cloud}}}{2\mu_{\text{cloud}}(1-\rho_{\text{cloud}})} \cdot \frac{c^2_{A_{\text{cloud}}}+c^2}{2k}
$\label{lemmagg1}
(Proof in Appendix \ref{AppGG1}).
\end{lemma}
Through the same methodology as in section  \ref{MMk Latency bound}, we extended our $M/M/k$ analysis to $GI/G/k$ in Appendix \ref{AppGG1}. The lower bound for $\Delta t$ for a $GI/G/k$ system is,
\begin{equation}
\Delta t > [\underbrace{
\frac{\lambda \left( \frac{1}{\mu_1^2} + \frac{r}{\mu_2^2} + \frac{r}{\mu_1 \mu_2} \right)}{1 - \frac{\lambda}{\mu_1} - \frac{r \lambda}{\mu_2}}
+ \frac{r\lambda}{\mu_1 (\mu_1- r\lambda)}] \cdot  
\frac{c_A^2 + c_S^2}{2}}_{w_{\text{edge}}}
 + \underbrace{\frac{r}{\mu_2}}_{s_{\text{migration}}} 
- \underbrace{
\frac{\rho_{\text{cloud}}^k + \rho_{\text{cloud}}}{2\mu_{\text{cloud}}(1-\rho_{\text{cloud}})} \cdot \frac{c^2_{A_{\text{cloud}}}+c^2_{S_{\text{cloud}}}}{2k}
}_{w_{\text{cloud}}}, \label{deltat_gg1_final1}
\end{equation}
   
Corollaries~\ref{corollary43} and~\ref{corollary2} also hold for the $GI/G/k$ case. Similar to Corollary~\ref{corollary3}, a lower bound on $t_{\text{cloud}}$ can be derived for the $GI/G/k$ system as well. Lemma \ref{lemmagg1} becomes particularly useful for deriving corollary \ref{ca2cor}



\begin{corollary}
The maximum workload coefficient of variation that can be handled at the edge cloud $c_{edge_A}^2$  has an upper bound limit of $\frac{2(\Delta t  - s_{migration}+w_{cloud})(1- \rho_{edge})}{\lambda ( \frac{1}{\mu_1^2} + \frac{r}{\mu_2^2} + \frac{r}{\mu_1 \mu_2})} - c_{S_{edge}}^2$ for the edge cloud performance to be better than a centralized cloud.
\label{ca2cor}
\end{corollary}

\noindent\textit{Proof: } 
Assuming the best case scenario from Corollary \ref{corollary43}, where $\rho_{dest} \rightarrow 0$, and therefore have no impact on the waiting time, then, manipulating equation \eqref{deltat_gg1_final1}, we get,
\begin{equation}
  c_{A_{edge}}^2 < \frac{2(\Delta t - s_{migration} + w_{cloud}) \overbrace{\left(1 - \frac{\lambda}{\mu_1} - \frac{r \lambda}{\mu_2} \right)}^{\text{remaining capacity of edge i.e. $(1- \rho_{edge})$}}
}{\lambda ( \frac{1}{\mu_1^2} + \frac{r}{\mu_2^2} + \frac{r}{\mu_1 \mu_2})} - c_{S_{edge}}^2
\end{equation}

This is an upper limit on the maximum workload coefficient of variation that an edge cloud can handle while still maintaining lower latencies compared to a remote cloud. As the edge utilization increases and approaches 1, the edge cloud cannot serve all requests in a workload spike and needs to drop or delay many of these requests in the spike. This is not surprising and is intuitive. However, an important factor that should be considered when deciding if a spike should be handled on the local edge cloud or should be offloaded is the latency difference between the edge cloud and the remote cloud.

\noindent\textbf{Implications: }
It is not sufficient to provision resources based on average response times. Systems must be dimensioned for tail performance, as even modest increases in arrival or service variability can significantly inflate delays, as characterized by the term $(c_A^2 + c_S^2)/2$, which directly scales waiting times beyond utilization effects. Planners can use the derived bound to compute a worst-case coefficient of variation and size buffers and servers for 95th or 99th percentile response times. 
For example, users in the same vehicular convoy or train car often hand off together, causing synchronized spikes in arrival streams that violate Poisson assumptions. Isolating traffic from these synchronized handoffs into a dedicated scheduling queue or resource pool, or throttling their migration rate $r$, could help prevent accidentally creating a 50‑session bursts at the edge.
By tracking the variance-induced latency penalty, operators can decide when to offload workloads dynamically, while SLA teams may set stricter admission control thresholds to avoid sudden spikes.
\subsection{\textbf{Model Validation}} \label{sec:validation}
\begin{figure*}
   \begin{minipage}{0.66\textwidth}
   \begin{subfigure}{0.45\columnwidth}
        \centering
        \includegraphics[width=\linewidth]{ 
     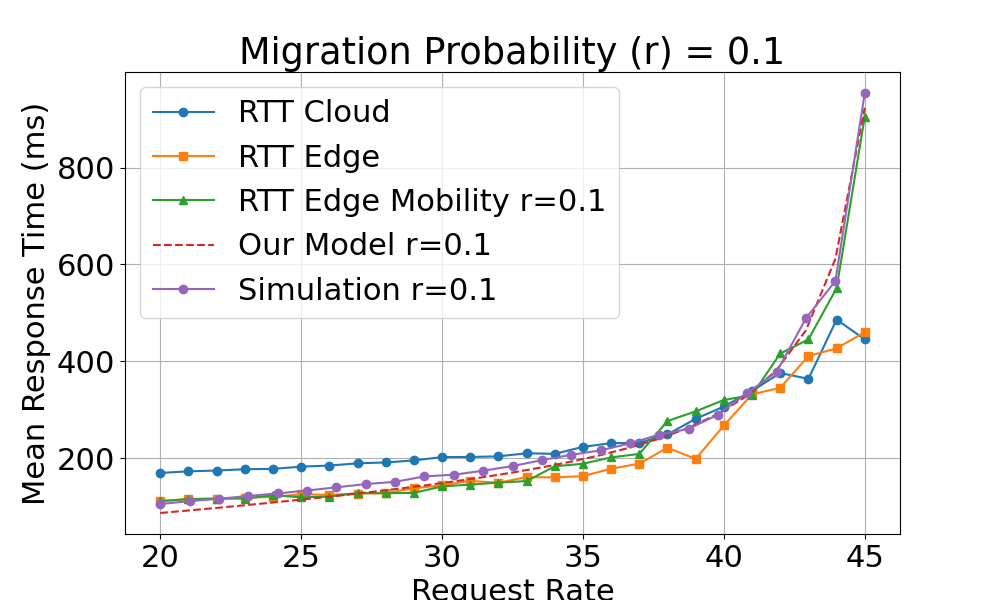}
        \caption{Mobility overhead r=0.1}
        \label{fig:mob}
    \end{subfigure}
     \begin{subfigure}{0.45\columnwidth}
        \centering
        \includegraphics[width=\linewidth]{ 
     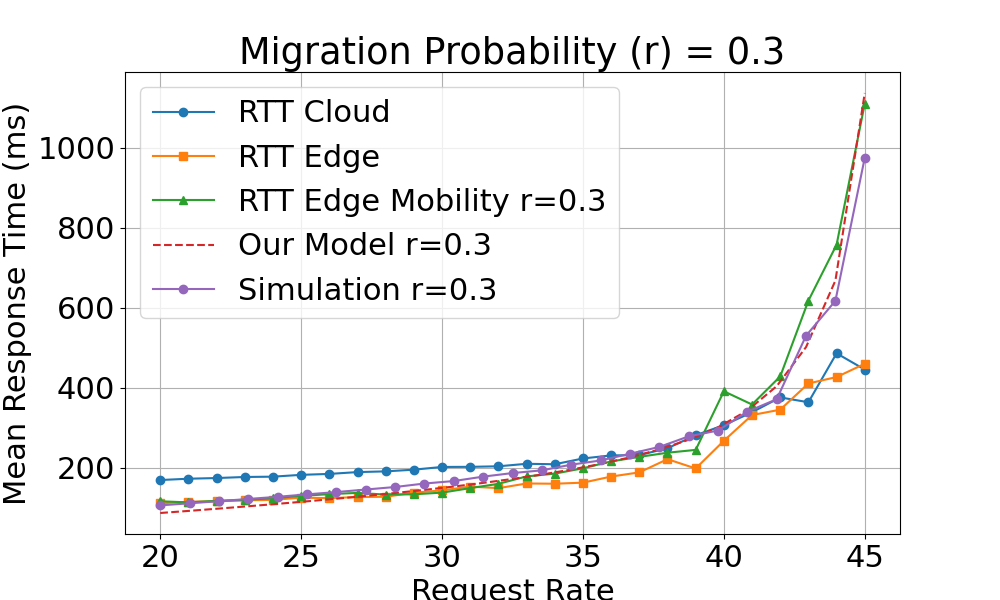}
        \caption{Mobility overhead r=0.3}
        \label{fig:mob1}
    \end{subfigure}    
    \caption{Response times with mobility and migration overhead \\
    for different migration probabilities compared with our $(M/M/k)$ \\
    model. ($\mu_1 = \mu_2= 50$)}
    \label{fig:mobmig}
\end{minipage}
   \begin{minipage}{0.33\textwidth}
        \centering
        \includegraphics[width=\linewidth]{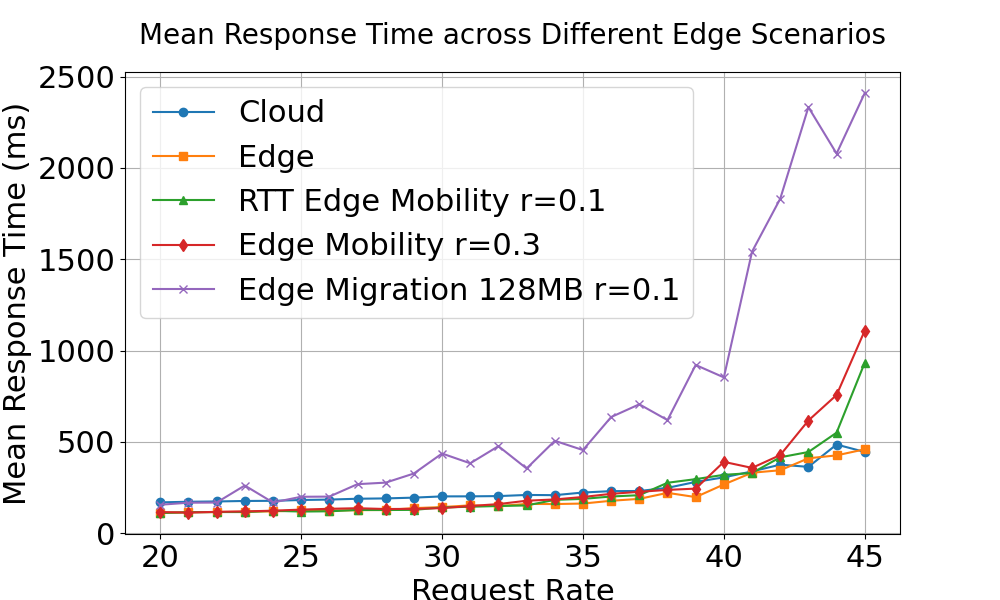}
        \caption{Migration of a 128 MB context}
        \label{fig:mig}
        \end{minipage}
    \end{figure*}
\noindent{\textbf{Experimental Setup}}  
We used Google Cloud ND2 servers (AMD EPYC, 4 virtual cores, 4 GB RAM). In the edge scenario, clients and servers are hosted within the same region (us-central) with approximately 1 ms latency. In the cloud scenario, the client is placed in a distant region (us-east4) with a latency of around 28 ms. The application is a Flask-based image classification service running a pre-trained ResNet50 model. A load generator simulates traffic by sending requests at a specified rate, following an exponential inter-arrival time distribution. We simulate mobility overhead through two mechanisms: (1) \textbf{Client Mobility: } where clients switch between edge sites with a migration probability of \(r = 0.1\) or \(r = 0.3\), forwarding requests to a second edge site upon migration; and (2) \textbf{Context Migration: } where a 128 MB session file is transferred between edge sites during migration, increasing the migration service time.

\noindent{\textbf{Results}}  
Figures~\ref{fig:mobmig} show the mean response times across different request rates in the cloud and edge scenarios. Without mobility, edge computing outperforms the cloud due to lower latency and proximity. However, as mobility overhead is introduced (Figures~\ref{fig:mobmig}a and~\ref{fig:mobmig}b), this advantage diminishes. Frequent client migrations lead to request forwarding and reprocessing delays, while context migration (Figure~\ref{fig:mig}) further increases service time and degrades performance.

We also validated our model with simulations with two M/M/1 queues: requests arrive at Queue 1 (edge server) following a Poisson process, and a fraction \(r\) migrates to Queue 2, incurring a delay \(\mu_2\). The utilization of Queue 2 (\(\rho_{dest}\)) increases with \(r\), which in turn affects waiting times in Queue 1. Results are averaged over 30 runs. Figures~\ref{fig:mobmig}a and~\ref{fig:mobmig}b show that simulated response times closely match empirical data, validating the model. Higher migration rates lead to queue congestion and increased response times, accurately capturing system behavior under varying workloads.

\section{Workload Skews}

In this section, we shift our focus to RQ2. Load spikes are among the major causes of performance degradation and service outages in networked systems. 
In this section, we evaluate analytically how load spikes are handled by edge versus remote clouds.
Our aim is to model how the increased number of servers allows a centralized cloud to handle more spikes compared to edge clouds, lowering the probability of delay for bursty workloads.
We assume that there are three types of edge sites, receiving different levels of arrival rates: one with a high demand, another with moderate demand, and a third with low demand.
We assume that an edge site experiencing a high load cannot load-balance the load with other edge sites, i.e., no migrations occur for balancing but only for user-mobility, and that each location operates independently. This imbalance results in increased queuing delays for users at high-load edge sites, as those requests must wait for the single-server capacity available there, even though other edge sites may have idle resources.
\subsection{Time-Varying skews overhead}
To model the dynamics of temporal workload skews across edge locations, we will use time-varying arrival queues. Each edge site will experience a time-dependent arrival rate, represented as $\lambda_{edge}(t)$. 
We again assume, an edge cloud is modeled using a $M_t/M/1$ queue and that there are $k$ distributed edge cloud locations. We also assume that a centralized cloud hosts an equivalent number of servers to the edge cloud, i.e., it has a total of $k$ servers, and thus can be modeled as a $M_t/M/k$ queuing system. The aggregated arrival rate at the cloud, $\lambda_{\text{cloud}}(t)$, is computed as the sum of the arrival rates at all $k$ edge sites:
\begin{equation}
 \lambda_{\text{cloud}}(t) = \sum_{i=1}^k \lambda_{edge(i)}(t).\label{eq:lambda_cloud}
\end{equation}
\begin{lemma}    
Assuming that the edge and cloud data centers are modeled as $M_t/M/1$ and $M_t/M/k$
queuing systems, respectively. In case of a non-stationary workload process, 
a tight lower bound for $\Delta t$ is $ \frac{\lambda(\mu_2^2 +r\mu_1^2+r\mu_1\mu_2)}{\mu_1\mu_2(\mu_
1\mu_2-\lambda\mu_2-r\lambda\mu_1)} + \frac{r\lambda}{\mu_1 (\mu_1- r\lambda)} + \frac{r}{\mu_{2}} - \frac{1}{\sqrt{k}\mu_{cloud}(1-\rho_{cloud}(t))}$
\end{lemma} 
\begin{proof}
Ross's Conjecture \cite{ross1978average} states and Rolski \cite{rolski1981queues} proves that, for nonstationary $(M_t/G/1)$ queue, having arbitrary service times $S$ and arrivals following a non-stationary Poisson process, the average delay \(d\) is greater than in a stationary system with the same mean arrival and service rates. Ross provided a lower bound for the average delay, given by the inequality:
\begin{equation}
d \geq \frac{\bar\lambda  \mathbb{E}[S^2]}{2(1 - \rho)},\label{eq:ross_conjecture}
\end{equation}
Where right-hand side is the delay of a stationary M/G/1 system. Furthermore, \cite{rolski1981queues} demonstrates that greater non-stationarity in the arrival process leads to increased average delay. This indicates that time-varying workloads can introduce additional delays at the edge, which must be accounted for in system design and analysis. Hence, the stationary delay estimate is a lower bound for the non-stationary queue.


To model the waiting time of a cloud system with $k$ servers under time-varying arrivals (i.e., an $M_t/M/k$ queue), we employ the \textit{Pointwise Stationary Approximation (PSA)} \cite{whitt1991pointwise}. PSA estimates the behavior of a non-stationary queueing system by assuming that, at each time instant, the system behaves as if it were in a stationary state with the instantaneous arrival rate $\lambda(t)$. In other words, at time $t$ one uses the $M/M/k$ formula as if the arrival rate were fixed at $\lambda(t)$. 

Green et al.~\cite{green1991pointwise} show that the accuracy of PSA improves with the number of servers, making it well suited for modeling cloud‐scale systems. However, the PSA provides an upper bound on the true time‐averaged waiting time because it assumes that, at every instant, the system behaves like a stationary queue with the current arrival rate. In reality, queues cannot adjust instantaneously to changes in input. PSA ignores transient dynamics and instead reacts immediately to every fluctuation, leading to an overestimate of the delay. Furthermore, analysis in \cite{grassmann1983convexity} shows that, under moderate variability, the steady‐state waiting time is a convex increasing function of the arrival rate. Convexity implies that short periods of high arrival rates contribute disproportionately more to delay than equally long periods of low rates reduce it. As a result, PSA over predicts the true average delay and thus acts as a tight upper bound.
Therefore, 
\[
\mathbb{E}(w)_{\text{stat}} < \mathbb{E}(w)_{non-stat} < \mathbb{E}(w)_{\text{PSA}}.
\]
If we keep the stationary estimate of $\mathbb{E}[w_{edge}]$ from eq. \eqref{eq:deltan6}, and change the $\mathbb{E}[w_{cloud}]$ to PSA, such that,
\begin{equation}
 \Delta t >   \frac{\lambda(\mu_2^2 +r\mu_1^2+r\mu_1\mu_2)}{\mu_1\mu_2(\mu_
1\mu_2-\lambda\mu_2-r\lambda\mu_1)} + \frac{r\lambda}{\mu_1 (\mu_1- r\lambda)} + \frac{r}{\mu_{2}} - \frac{1}{\sqrt{k}\mu_{cloud}(1-\rho_{cloud}(t))}, \quad \forall t > 0\label{eq:deltan8nonstationary}
\end{equation}
then \eqref{eq:deltan8nonstationary} provides a conservative lower bound for $\Delta t$ in case of a non-stationary workload process, because we are taking the best case for the edge and the worst case for the cloud.

This completes the proof.    
\end{proof}
\begin{lemma}
Assuming that the edge and cloud data centers are modeled as $M_t/M/1$ and $M_t/M/k$ systems, respectively, and the arrivals have a \textbf{spike}, edge cloud offers lower response times whenever the network RTT difference between edge and cloud is greater than $\frac{\rho^{2}A^{2}}{
          {2\,\mu_{eff}}(1-\rho)^{3}(1+(\frac{\gamma}{\mu_{eff}})^{2})} + w_{edge_{stat}}+s_{mig}-w_{cloud}$
 (Proof in Appendix \ref{appspike})  \label{lemaspike}
\end{lemma}
\subsection{Time-Varying Sinusoidal arrivals}
To model time-varying workloads, we take the example of sinusoidal arrivals. We note that we do the analysis for a system with a perfect sinusoidal workload, which is the simplest possible variation of a workload. The arrival rate at an edge site can be modeled as a sinusoidal function, representing time-varying demand patterns observed over a time period ranging from a few seconds to a 24-hour cycle. 
\[
\lambda_{\text{edge}}(t) = \bar{\lambda} \left(1 + A \sin(\gamma t)\right),
\]
where \( A \) is the relative amplitude (\( 0 \leq A \leq 1 \)), and \( \gamma \) is the frequency of oscillation.

The relative amplitude \( A \) indicates the maximum proportional deviation from the baseline \( \bar\lambda \), and characterizes the degree of non-stationarity in the input process. At any time \( t \), the arrival process is assumed to follow a Poisson distribution with mean \( \lambda_{\text{edge}}(t) \).

We recall from Section~\ref{MMk Latency bound} that each request first undergoes a mandatory phase with service rate \( \mu_1 \), followed by a migration phase with rate \( \mu_2 \) with probability \( r \).
To account for multi-stage service mechanism, we define an effective service rate, \( \mu_{\text{eff}} \), which represents the total processing capacity of the system. the total expected service time is then,
\begin{equation}
\mathbb{E}[S] = \frac{1}{\mu_1} + \frac{r}{\mu_2} \quad \Rightarrow \quad \mu_{\text{eff}} = \left( \frac{1}{\mu_1} + \frac{r}{\mu_2} \right)^{-1}= \frac{\mu_1\mu_2}{\mu_2 + r\mu_1}.
\label{eq:mueff}    
\end{equation}

With the total effective service rate defined as \(\mu_{\mathrm{eff}}\), accounting for both non‑migrating and migrating arrivals, we model the excess waiting time under a workload spike in Appendix \ref{appspike} as. 

\begin{equation}
\Delta w_q(A)
   =\frac{\rho^{2}A^{2}}{
          {2\,\mu_{eff}}(1-\rho)^{3}(1+(\frac{\gamma}{\mu_{eff}})^{2})}.  \label{excesswait}
\end{equation}
Equation \eqref{excesswait} is the extra waiting time that will be added to the stationary $E(w)$. Putting it in \eqref{eq:deltan2}, we get the final inequality.
Furthermore, if we assume the arrivals are sinusoidal and the capacity is allocated to handle the peak load, one can show that the average utilization will be $\bar\rho \leq 0.5$ (Corollary \ref{sinutilization} in Appendix).  
Moreover, if we assume that the peaks and troughs of the sinusoidal arrival rates at individual sites occur at different times. When these out-of-phase arrival rates are summed, the resulting aggregated rate at the cloud exhibits a smoothing effect (Corollary \ref{sinsumstatcloud} in Appendix). 

Analysis in \cite{green1991some} provides numerical confirmation and an extension of Ross's conjecture, showing that the expected delay is convex increasing in the relative amplitude of the arrival process. The experiments in \cite{green1991some} reveal that with a relative amplitude of just 0.33, the actual expected delay for non-stationary arrivals is more than double the stationary expected delay. When the relative amplitude increases to 1, the actual expected delay surpasses ten times the stationary estimate. This highlights the significant impact of arrival process variability on system performance, emphasizing the importance of accounting for non-stationarity in queueing analysis of edge and cloud.

\noindent\textbf{Implications: }
In equation \eqref{excesswait}, the quadratic dependence on amplitude $A$ implies that even modest fluctuations in the arrival rate can cause significant increases in delay. The dominant term $(1-\rho)^{-3}$ further highlights that systems operating at high utilization $\rho$ experience a cubic blow-up in delay, making the edge particularly sensitive to periodic bursts as $\rho \to 1$. Furthermore, the factor $(1+(\gamma/\mu)^{2})^{-1}$ indicates that when $\gamma \ll \mu_{\text{eff}}$, low-frequency oscillations cause much larger delays because the queue has time to accumulate backlog, whereas when $\gamma \gg \mu_{\text{eff}}$, high-frequency oscillations are naturally smoothed out by the service process, effectively averaging to the mean load. This explains why synchronized, low-frequency events, can create severe latency spikes. These insights suggest that variance mitigation techniques (e.g., traffic smoothing, predictive offloading, or migration throttling) are critical for edge systems.
\begin{lemma}
   Assuming that the edge and cloud data centers are modeled as $M_t/M/1$ and $M_t/M/k$
queuing systems, respectively and the arrivals are sinusoidal. During \textbf{rush hour}, if $\lambda_{max} > \mu$, edge cloud offers lower response times whenever the network RTT difference between edge and cloud is greater than \( \frac{ \left( (\bar{\lambda} - \mu_{\text{eff}})(\pi - 2\theta)
+ 2\bar{\lambda} A \cos(\theta) \right)}{\gamma \mu_{\text{eff}}}\). (Full proof in Appendix \ref{appb})
\end{lemma}
\noindent \textit{Proof summary:}
We assume:
\[
\bar\lambda = \frac{1}{T} \int_0^T \lambda(t) \, dt < \mu, \quad \bar\rho = \frac{\bar\lambda}{\mu} < 1,
\]
where \( T \) is the period of the sinusoidal input. Thus, \( \bar\lambda \) represents the average arrival rate over one cycle. Under this assumption, the system is expected to reach a periodic steady-state \cite{heyman1984asymptotic}, where its behavior repeats every \( T \) seconds.

The arrival rate \( \lambda(t) \) gradually increases as a "rush hour" approaches and eventually exceeds the server's service capacity \( \mu \). As the rush hour passes, the arrival rate decreases, and the queue that builds up during the peak period is gradually cleared before the next cycle. 
The maximum arrival rate,
\[
\lambda_{\max} = \bar{\lambda}(1 + A) > \mu.
\]
With the total service rate defined as \( \mu_{\text{eff}} \) \eqref{eq:mueff}, we analyze the queue's behavior during periods of overload using a \textit{fluid model} approximation \cite{newell1968queues,keller1982time}.  This approach treats the arrival and service processes as continuous flows, allowing us to describe the system dynamics using integrals instead of discrete-event analysis. The fluid model is particularly useful for analyzing overload conditions, where the arrival rate temporarily exceeds the service capacity, leading to queue buildup. By treating the system as a continuous flow, the model emphasizes on the cumulative impact of time-varying demand, while abstracting away stochastic fluctuations.
To compute the total backlog accumulated over the overload period \([t_1,t_2]\), we integrate the instantaneous excess:
\[
V(t_1,t_2) =\int_{t_1}^{t_2} \bigl(\lambda(t)-\mu_{\mathrm{eff}}\bigr)\,dt,
\]
where $t_1$ and $t_2$ are the start and end times of a time interval during which the arrival rate exceed the service rate, leading to potential backlog accumulation. We have derived $V(t_1, t_2)$ in Appendix \ref{appb} as,
\begin{equation}
V(t_1,t_2) = \frac{1}{\gamma} \left( (\bar{\lambda} - \mu_{\text{eff}})(\pi - 2\theta)
+ 2\bar{\lambda} A \cos(\theta) \right), \quad \text{where, } \quad \theta = \arcsin\left( \frac{\mu_{\text{eff}}}{\bar{\lambda}} - 1 \right) 
 \label{v}
\end{equation}

If the system is processing jobs at rate $\mu$
, and there is a backlog $Q(t)$, the average time it will take to serve them is:
\begin{equation}
E(w)_{rush} = \frac{ \left( (\bar{\lambda} - \mu_{\text{eff}})(\pi - 2\theta)
+ 2\bar{\lambda} A \cos(\theta) \right)}{\gamma \mu_{\text{eff}}}. \label{v1}
\end{equation}

In the stationary analysis of \( \mathbb{E}(w) \), the arrival rate \( \lambda \) must satisfy \( \lambda < \mu \) to ensure system stability. Under this assumption only, the stationary model provides a useful estimate for delay. However, in case of non-stationary arrivals, during rush hours, when the instantaneous arrival rate temporarily exceeds the service capacity, even brief overload periods can lead to a sharp and nonlinear increase in waiting times. 
To better approximate delays during such overloads, we use Equation~\eqref{v1} to quantify the backlog accumulated during these transient periods. 

Because $E(w)_{rush}$ dominates the behavior, we omitted additional terms ($s_{mig},w_{cloud}$). Therefore, during the rush hour if $\lambda_{max} > \mu$, then,
\begin{equation}
 \Delta t >   \frac{ (\bar{\lambda} - \mu_{\text{eff}})(\pi - 2\theta)
+ 2\bar{\lambda} A \cos(\theta)
}{\gamma \mu_{\text{eff}}} \label{eq:deltan8nonstationary2}
\end{equation}
\noindent\textbf{Implications: }
In Lemma \ref{lemaspike} we calculated the extra delay added to the expected waiting time. Equation \eqref{eq:deltan8nonstationary2} provides a picture of the time interval during which arrivals exceed service rates. The delay is not averaged over the low-arrival period but reflects the backlog users actually face during these overload intervals. The term $\pi - 2\theta$ represents the fraction of the sinusoidal cycle where $\lambda(t) > \mu_{\text{eff}}$, meaning that larger values of $\pi - 2\theta$ correspond to longer overload durations and thus more accumulated delay. The relative amplitude $A$ directly scales the additional delay contribution. As seen in Lemma \ref{lemaspike}, again, low-frequency spikes (small $\gamma$) are particularly damaging, as the system remains in overload for longer periods, while high-frequency oscillations are smoothed out by the service process. These results imply that edge systems must either provision sufficient headroom (larger $\mu_{\text{eff}}$) or implement traffic shaping and predictive offloading during rush-hour peaks, especially for bursty workloads such as AR/VR, interactive gaming, or ML inference.


We analysed how temporal variability in the arrival rates 
can affect the performance at the edge. The cloud system benefits from traffic aggregation, which smooth out fluctuations in the workload and brings the system closer to a stationary regime.

\subsection{Model Validation} \label{sec:validation}
 \begin{figure*}
   \begin{minipage}{.6\textwidth}
        \begin{subfigure}{0.5\columnwidth}
        \centering
        \includegraphics[width=\linewidth]{ 
     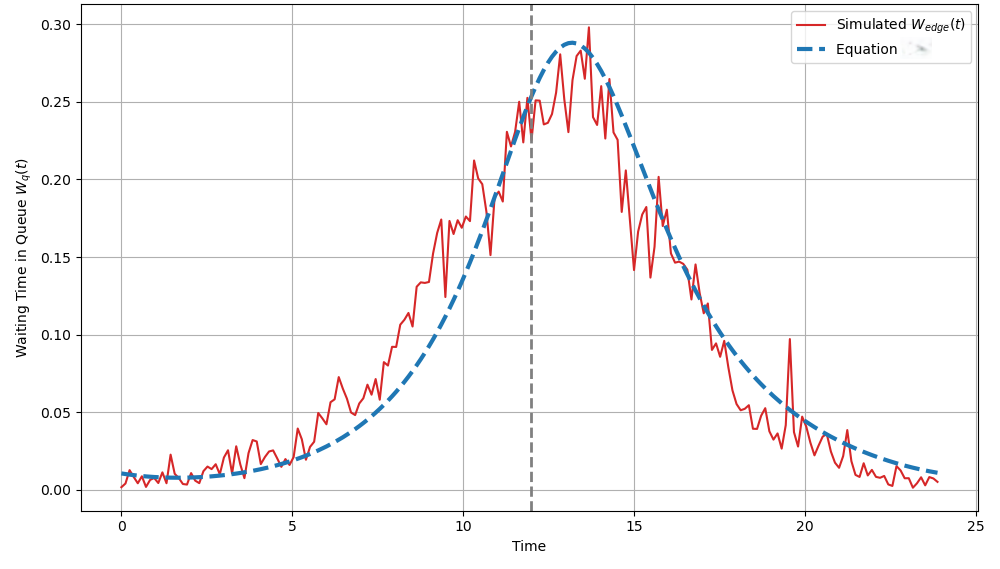}
        \caption{Simulation vs Equation  \eqref{eq:wedgesin} $\mu_{eff} = 20, \bar\lambda=10$, A = 0.7}
        \label{fig:sin07}
    \end{subfigure}
    \begin{subfigure}{0.45\columnwidth}
        \centering
        \includegraphics[width=\linewidth]{ 
     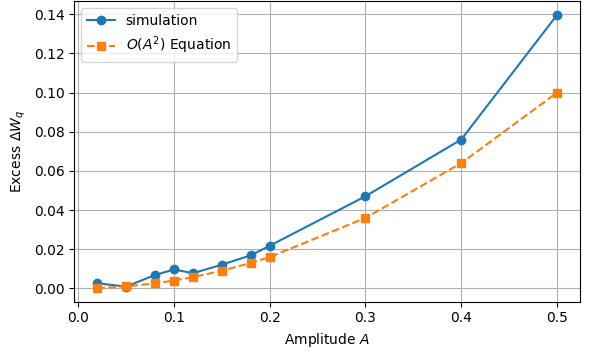}
        \caption{Simulation vs Equation  \eqref{excesswait} $\mu_{eff} = 100, \bar\lambda=80$}
        \label{fig:sin03}
    \end{subfigure}
       \caption{Simulation waiting times with sinusoidal arrivals 
      compared with Equation  \eqref{eq:wedgesin} and \eqref{excesswait}.}
    \label{fig:sineq28}
\end{minipage}
   \begin{minipage}{0.3\textwidth}
        \centering
        \includegraphics[width=\linewidth]{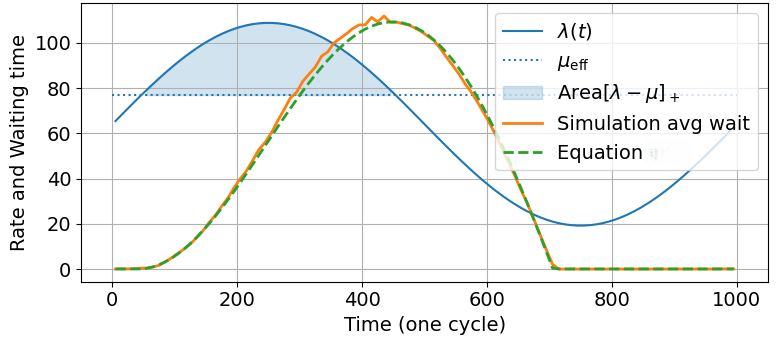}
        \caption{Simulation waiting times with sinusoidal arrivals during rush hour 
      compared with equation \eqref{v1}. ($\mu_1 = \mu_2= 100, \bar\lambda=64$, A=0.7)}
        \label{fig:sinresult}
        \end{minipage}
    \end{figure*}

\textbf{Excess waiting time of a sinusoidal spike}:
We simulated an \( M_t/M/1 \) queue and averaged the waiting times over 30 runs. Figure~\ref{fig:sin07} compares the simulated waiting time with the analytic prediction from Equation~\eqref{eq:wedgesin}. The model captures the shape and timing of the delay well, though it slightly underestimates the peak. A noticeable lag in the waiting time also appears at higher amplitude (\(A = 0.7\)).

Figure~\ref{fig:sin03} compares the excess waiting time from the analytic model (equation~\eqref{excesswait}) with simulation results. Since the model includes only the second-order term, it slightly underestimates the simulated excess waiting time, which is computed by subtracting the stationary \( M/M/1 \) waiting time from the simulated average waiting time.

\noindent\textbf{Fluid model (\(\lambda_{max} > \mu_{eff}\))}:
We track the average waiting time \( w(t) \) during the overload period. Figure~\ref{fig:sinresult} compares the simulation results with the time-varying delay approximation from equation ~\eqref{v1}. In Figure~\ref{fig:sinresult} \( w(t) \) follows the sinusoidal pattern with a phase lag due to queueing inertia.
\begin{table}[tb]
 \caption{Waiting times during rush hour under different amplitudes \( A \), based on simulation and equation~\eqref{v1}.}
 \footnotesize
  \label{tab:tb1}
  \centering
  \begin{minipage}{0.48\linewidth}
    \centering
    \caption*{$\mu_1 = \mu_2 = 32, \quad \bar\lambda=16$}
    \begin{tabular}{@{}rcccc@{}}
      \toprule
      $A$ & $E(w)_{\text{sim}}$ & $E(w)_{\text{sim}}^{\text{rush}}$ & $E(w)_{f}^{\text{rush}}$ & $\mathrm{Err}_{\text{rush}}$ \\
      \midrule
      0.1 & 0.0773 & 0.0000 & 0.0000 & 0.0000 \\
      0.2 & 0.0895 & 0.0000 & 0.0000 & 0.0000 \\
      0.3 & 0.1087 & 0.0000 & 0.0000 & 0.0000 \\
      0.4 & 0.1753 & 0.0000 & 0.0000 & 0.0000 \\
      0.5 & 0.3206 & 0.0000 & 0.0000 & 0.0000 \\
      0.6 & 1.7501 & 4.3048 & 2.0082 & 2.2966 \\
      0.7 & 5.2432 & 9.1399 & 7.8479 & 1.2919 \\
      0.8 & 11.7600 & 16.9049 & 15.1707 & 1.7341 \\
      0.9 & 19.8773 & 25.3995 & 23.9695 & 1.4299 \\
      1.0 & 30.4096 & 35.0717 & 31.9333 & 3.1384 \\
      \bottomrule
    \end{tabular}
  \end{minipage}
  \hfill
  \begin{minipage}{0.48\linewidth}
    \centering
    \caption*{$\mu_1= \mu_2 =512, \quad \bar\lambda=256$}
    \begin{tabular}{@{}rcccc@{}}
      \toprule
      $A$ & $E(w)_{\text{sim}}$ & $E(w)_{\text{rush}}^{\text{sim}}$ & $E(w)_{rush}^{\textit{f}}$ & $\mathrm{Err}_{\text{rush}}$ \\
      \midrule
      0.1 & 0.0049 & 0.0000 & 0.0000 & 0.0000 \\
      0.2 & 0.0055 & 0.0000 & 0.0000 & 0.0000 \\
      0.3 & 0.0070 & 0.0000 & 0.0000 & 0.0000 \\
      0.4 & 0.0103 & 0.0000 & 0.0000 & 0.0000 \\
      0.5 & 0.0240 & 0.0000 & 0.0000 & 0.0000 \\
      0.6 & 0.9741 & 2.5473 & 2.0082 & 0.5391 \\
      0.7 & 4.9876 & 8.6203 & 7.8479 & 0.7724 \\
      0.8 & 11.1964 & 16.1445 & 15.1707 & 0.9738 \\
      0.9 & 19.1184 & 24.5659 & 23.9695 & 0.5963 \\
      1.0 & 29.6241 & 34.4553 & 31.9333 & 2.5220 \\
      \bottomrule
    \end{tabular}
  \end{minipage}
\end{table}

Table~\ref{tab:tb1} compares the simulated waiting times \( E(w)_{\text{sim}}^{rush} \) with the analytic estimate \( E(w)_f^{rush} \) for amplitudes \( A \in [0.1, 1.0] \). When \( A \) is small, the system experiences no overload, and delays remain low. As \( A \) increases beyond 0.6, delay increases significantly. In this regime, the fluid rush-hour approximation \( E(w)_f^{\text{rush}} \) serves as a lower bound and provides a closer approximation. The approximation error decreases as the service rate \( \mu \) increases.

\section{Edge Cloud Capacity Tradeoffs}
In this Section, we shift our focus to trying to estimate how much extra capacity is needed per edge site to service requests such that the edge site performance is similar or better than an edge cloud. Here, we add to our model an aspect that has not been studied before, to the best of our knowledge; How does hosting virtualized services on both the edge and the cloud affect the capacity requirements. In this scenario, an edge site can host multiple virtualized computations, e.g., Virtual Machines or containers. 

\subsection{Modeling Edge Capacity}
Consider a distributed edge cloud with $k$ edge locations and an aggregate server capacity of $C_{edge}$ across the system versus a centralized cloud serving the total workload of the distributed edge locations using a single data center with a server capacity of $C_{cloud}$. We assume that servers in the edge and centralized cloud are virtualized and resources are allocated to applications in the form of virtual machines (VMs). We assume that VMs arrive for scheduling according to a Poisson process with rate $\lambda$.  In the distributed edge cloud, a new VM request is assumed to be directed to the nearest edge cloud location (i.e., the one with the least distance $d$ from the user).   Requests are assumed to be uniformly distributed across the $k$ edge sites. We assume that a customer specifies a VM size (e.g., in terms of number of CPU cores) along with each VM request. This is inline with cloud platforms that support virtual machines of different sizes. The VM size is assumed to be generally distributed. We assume $q$ to be lower bound of the number of virtual machines that can be hosted at each edge cloud site. Thus, $q$ represents the VM {\em packing factor} that is a function of how many VMs can be  packed onto each server and the maximum VM size that can be requested on a server.  Finally, we assume the lifetime of a VM is $\zeta$, which represents the time between a customer requesting a VM and terminating it; from a queuing standpoint $\zeta$ represents the service time of a request. Then, $\bar \zeta$ and $\bar \zeta^2$ denote the mean and variance of the service times (i.e., VM lifetimes).  The utilization of each edge cloud site is given by  $\rho=\lambda \bar \zeta/k$.

We note that the problem of multiplexing virtual machines on virtualized cloud servers has been modeled as an online bin-packing problem~\cite{wang2011consolidating,gupta2012online} and many VM placement policies based on bin packing have been proposed. Classical bin-packing is a well-studied combinatorial optimization problem 
\cite{coffman1999bin}, with some algorithms allowing requests to be queued until they are placed in a bin~\cite{shah2008bin}. However, the classical model for bin-packing is not suitable for modeling an edge cloud, since in our case the arrivals of items to place in the bin is localized to a certain edge, and thus cannot be placed on any other edge.

In order to capture these constraints, we model the edge cloud scenario as a Dynamic Traveling Repairman Problem (DTRP), a model while popular in fields such as logistics, but has not found much application in modeling computer systems.
Given the above system model, we have the following lemma for the distributed edge cloud capacity
\begin{lemma}
To give the same performance, the capacity of the centralized cloud can be calculated as a function of edge cloud as follows,
 \newline
  $C_{cloud} = C_{edge} \frac{(1-\rho_{edge}-\frac{\tau_{edge}}{C_{edge}})}{(1+1/q_{edge})(1-\rho_{cloud})}$
\end{lemma}
\noindent\textbf{Proof:}
VM placement in distributed edge clouds is similar to bin packing where incoming VM requests have a geographic affinity to the nearest edge cloud site.  To model capacity in such a scenario, we use a result
from the classical work of Bertsimas and Van Ryzin~\cite{bertsimas1993stochastic} where results from geometrical probability, queuing theory, and combinatorial optimization were used to obtain results on the 
{\em Dynamic Traveling Repairman Problem (DTRP)}.  In DTRP, repair requests arrive at a particular geographic locations. There are $k$  cars with repair persons that travel with velocity $v$ to the locations of repair requests, and each repair request is satisfied by the nearest car. In our case, a car (i.e., repair person) is equivalent to a
distributed edge cloud location having latency $v$, that provides service, while a repair request is equivalent to a VM; similar to a repair request being serviced by the nearest traveling repair-person (car), VM requests are satisfied by the nearest distributed edge cloud. The DTRP analysis in  ~\cite{bertsimas1993stochastic} address the bin packing problem of how repair requests are mapped onto traveling repair persons (``bins''). The  response time $T$ (referred to as system time in ~\cite{bertsimas1993stochastic}) is the total time from the initiation of a VM request, i.e., when the user starts the upload process of the VM image to the edge site, to the time it departs (i.e., is terminated by the customer).  The analysis in ~\cite{bertsimas1993stochastic} provides packing policies for the DTRP problem with constant factor guarantees such that
\begin{equation}
    T \sim \gamma^2 \frac{\lambda A(1+1/q)^2}{C^2 v^2 (1-\rho-2\lambda \bar U/Cq)^2},  \label{eq:dtrpPack}
\end{equation}
as $\rho+\frac{2 \lambda \bar d}{v q} \longrightarrow 1$, i.e., at higher utilization and assuming $T$ is the finite response time of the packing policy (implying that VM lifetimes need to be finite). 
In this expression, $A$ is the area of a Euclidean
 service region, $U= \frac{\bar d}{v}$ is the mean time taken to upload a VM to the edge, and  $\gamma$ is a constant representing the grade of service. 
We assume that the  edge  cloud and the centralized cloud use identical packing policy, as described in~\cite{bertsimas1993stochastic}, with finite response times $T$. In this case, the utilization 
$\rho=\lambda \bar \zeta/s$, while $\tau_{edge}=2 \lambda \bar d/qv$, is the total expected upload time of the VM to the edge or the cloud.
For the edge cloud, Equation~\ref{eq:dtrpPack} becomes,
\begin{equation}
    T_{edge} \sim \gamma^2 \frac{\lambda A(1+1/q_{edge})^2}{C_{edge}^2 v^2 (1-\rho_{edge}-\frac{\tau_{edge}}{C_{edge}})^2},  \label{eq:edgeTu}
\end{equation}

For a centralized  cloud, with capacity $C_{cloud}$, 
the number of VMs that can be packed on cloud data center before it becomes fully utilized, is very large, and is at least equal to $q_{edge}*C_{edge}$, making both $1/q$, and $\frac{\tau_{cloud}}{C_{cloud}}$ negligible. For the cloud, Equation~\ref{eq:dtrpPack} becomes,
\begin{equation}
    T_{cloud} \sim \gamma^2 \frac{\lambda A}{C_{cloud}^2 v^2 (1-\rho_{cloud})^2}. \label{eq:remTu}
\end{equation}

If we wish to provision edge and cloud capacity such that both the distributed edge cloud and centralized 
cloud give the same packing performance under the same packing policy, 
 Equations~\ref{eq:edgeTu} and Equation~\ref{eq:remTu} need to be equal. Since both systems aim to serve the same area and the same request arrival rates, we obtain
\begin{equation}
    C_{cloud}=\frac{C_{edge}(1-\rho_{edge}-\frac{\tau_{edge}}{C_{edge}})}{(1+1/q_{edge})(1-\rho_{cloud})}. \label{eq:remTu3}
\end{equation}
\begin{corollary} \label{corr:q2} 
The capacity of the centralized cloud  needed to service the same VM workload is less than 
that of the distributed edge cloud, that is 
$C_{cloud} < C_{edge}$
\end{corollary}
\noindent\textbf{Proof:}
Since $1-\rho_{edge}-\frac{\tau_{edge}}{C_{edge}}$ is necessarily less than 1 (and greater than 0), $1-\rho_{edge}-\frac{\tau_{edge}}{C_{edge}} \leq 1-\rho_{cloud}$ for equivalent server utilization, and $1+1/q_{edge}$ is necessarily greater than 1, then it follows that $C_{cloud}<C_{edge}$. 

\begin{corollary} \label{corr:q}For an edge cloud data centers with a large capacity $C_{edge}$, and when the system utilization at both the edge and remote clouds are equal, $\rho_{edge}=\rho_{cloud}$, then,
\begin{equation}
    \frac{C_{edge}}{C_{cloud}} = (1+1/q_{edge}). \label{eq:remTu2}
\end{equation}
\end{corollary}
\noindent\textbf{Proof:} Let the term $C_{edge}/C_{cloud}$ represent the scaling factor of the number of servers required at the edge to provide similar packing capabilities at the edge. The term $0<\rho_{edge}-\frac{\tau_{edge}}{C_{edge}}\leq1$, and if $C_{edge}$ is large enough, then $\tau_{edge}/C_{edge}\longrightarrow0$. The above corollary can therefore  be trivially proven.

\subsection{Model validation}
\noindent \textbf{Validation Setup:} In order to validate the above result, we use a real VM trace from Azure cloud services~\cite{cortez2017resource}. The dataset contains data on the arrival times, VM utilization, and sizes of over 2 Million VMs from Azure data centers in 2019. The average size of a VM is 4.75 cores, with the largest VM having 20 cores, and the smallest having 2 cores. Using the VM traces as a workload, we build a simulator that simulates edge setups of different sizes w.r.t. number of cores and/or servers per edge. The simulator is also capable of simulating remote large-scale clouds with different configurations. The simulator enables us to simulate edge sites with single, and with multiple servers, e.g., a small cluster of 10 servers with 64 cores each. This enables us to simulate different possible edge architectures, for example, some edge cloud designs suggest that each edge site should have a few racks~\cite{hilt2019future}, while others suggest an edge of a few or even a single server~\cite{choy2014hybrid, yin2016edge}.

\noindent \textbf{Results:}
The figures \ref{fig:nu10} and \ref{fig:ali-mem-band}  show how the edge infrastructure requires a greater number of servers compared to the cloud to maintain similar packing capabilities. Figure \ref{fig:nu10}  shows the scaling factor, which indicates the ratio of servers needed at the edge versus the cloud to achieve equivalent workload efficiency. When this factor is greater than 1, it means that more edge servers are necessary to match cloud performance. The plot reveals that at lower core counts per server, the edge requires significantly more servers, but as core densities increase, the scaling factor decreases, approaching 1.
Figures \ref{fig:ali-mem-band} compare different edge core configurations to cloud setups. As the edge core count nears the model’s predicted optimal values (96 cores/edge for 64-core cloud servers and 160 cores/edge for 128-core cloud servers), the error decreases significantly. This shows that, the edge requires more cores than the cloud to maintain comparable packing performance. 

 \begin{figure*}
    \begin{minipage}{0.3\textwidth}
    \centering
    \includegraphics[width=\textwidth]{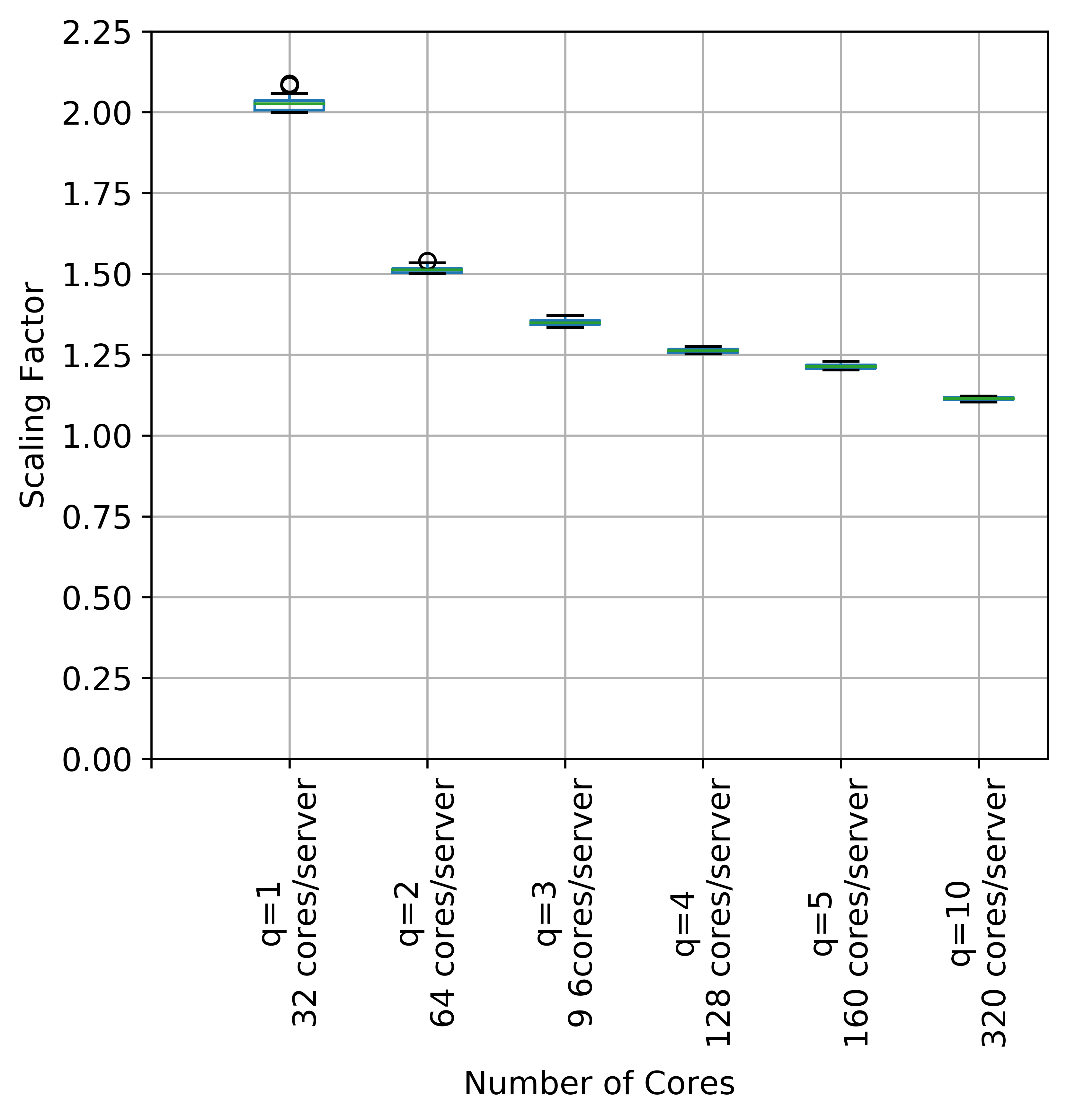}
    \caption{The model accurately predicts the scaling factor required between edge capacity and cloud capacity.}
    \label{fig:nu10}
\end{minipage}
   \begin{minipage}{0.66\textwidth}
    \centering
    \begin{subfigure}{0.43\textwidth}
        \centering
        \includegraphics[width=\linewidth]{ 
     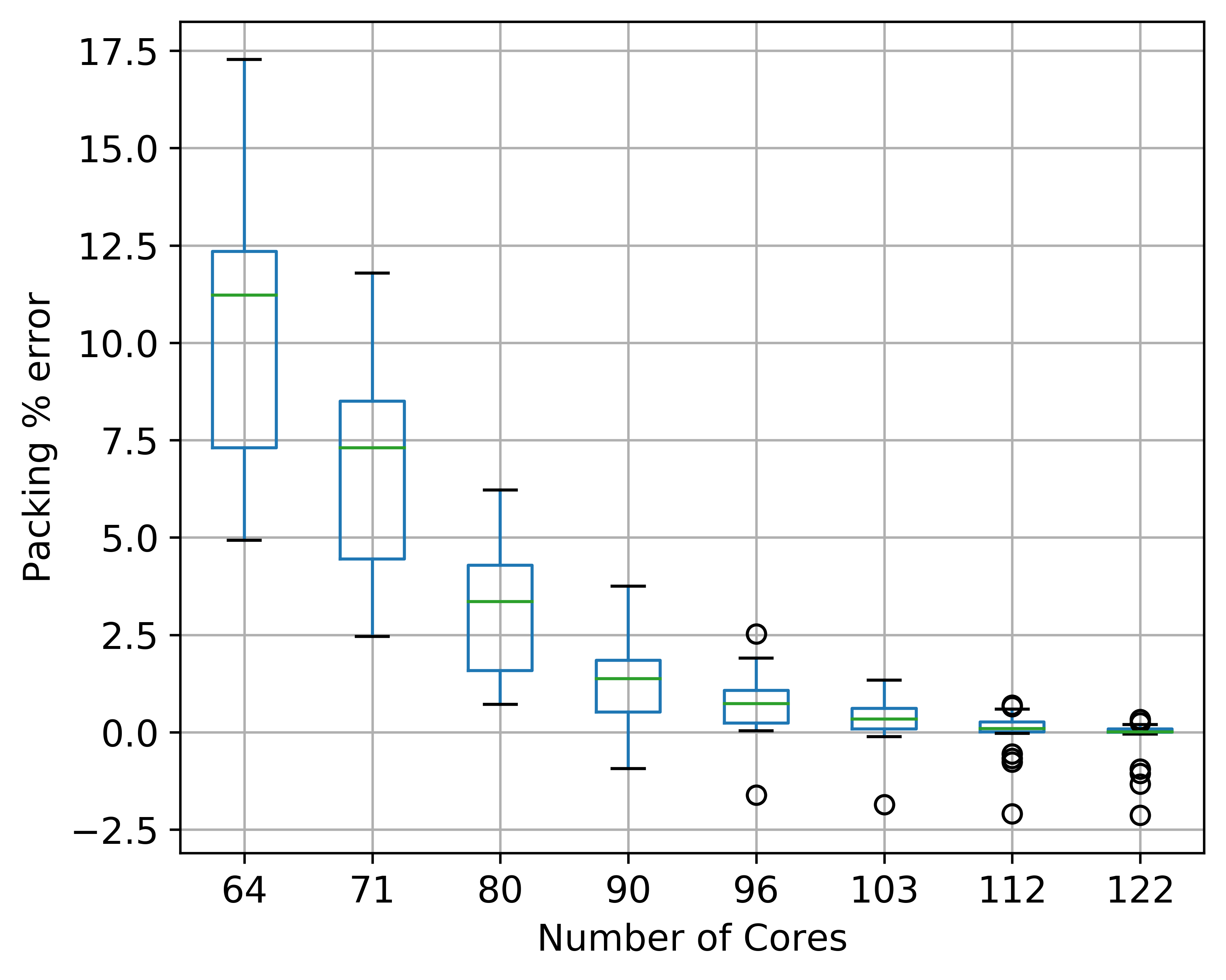}
        \caption{Relative error for various edge configurations (cores/edge) compared to a cloud with 1000 64-core servers (Model computes 96 cores/edge).}
        \label{fig:wikiWL}
    \end{subfigure}
    \begin{subfigure}{0.42\textwidth}
        \centering
        \includegraphics[width=\linewidth]{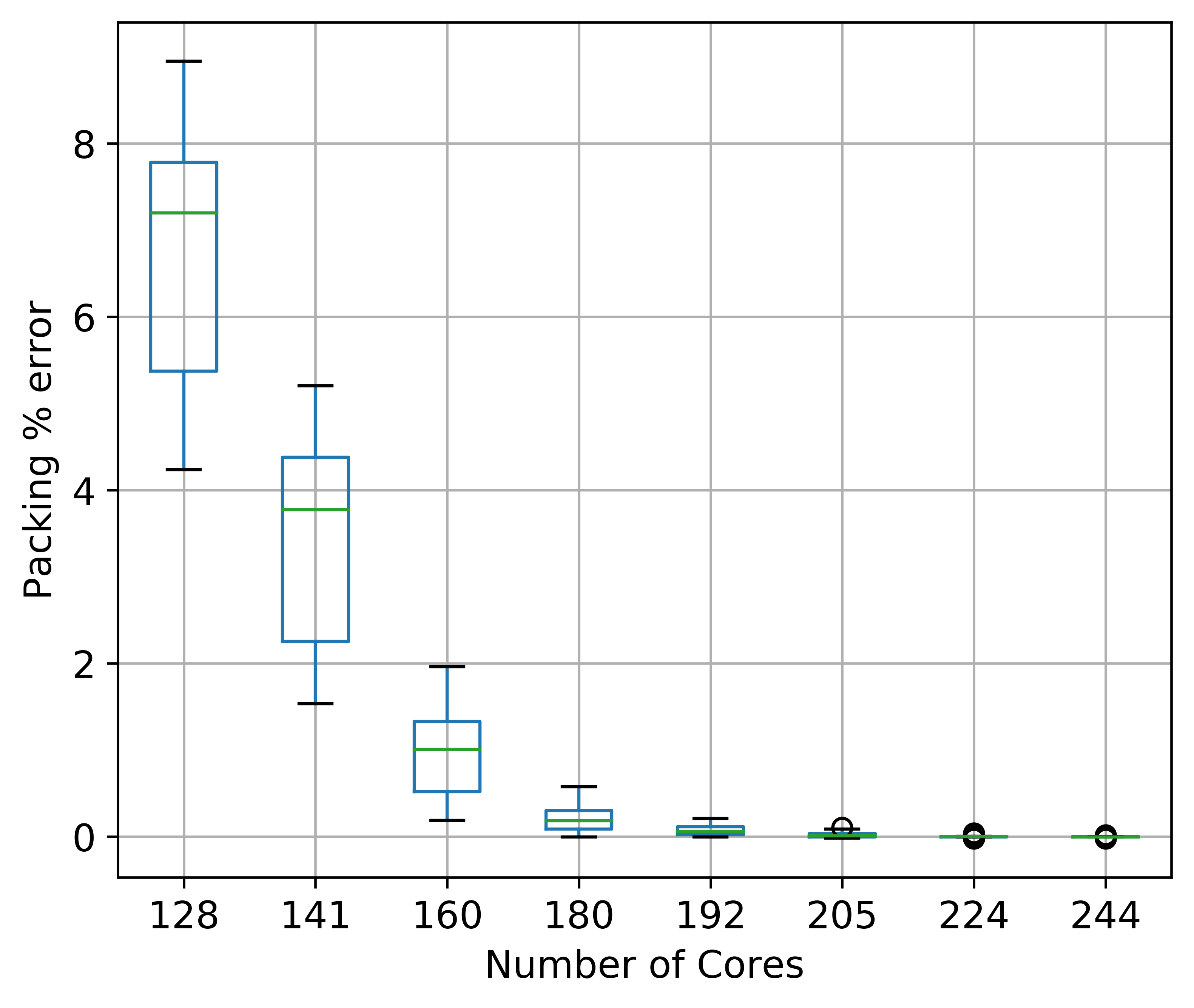}
        \caption{Relative error for various edge configurations (core/edge compared to a cloud with 1000 128-core servers (Model computes 160 cores/edge).}
        \label{fig:wikiWL2}
    \end{subfigure}
    \caption{There are diminishing returns w.r.t. the relative error compared to cloud allocation when the extra capacity is greater than what our model predicts.}
    \label{fig:ali-mem-band}
    \end{minipage}
\end{figure*}

\section{Further Implications}
Edge computing has recently been marketed as a one-size fits all solution to privacy, performance, and the next-generation of autonomous systems. However, as our models and simple experiments reveal, the reality is more complicated. 

\noindent \textbf{Placement Implications:} 
\label{sec:placement}
Since centralized clouds platforms such as Amazon EC2 and Azure have begun to
deploy additional data centers in various geographic regions, the latency to centralized clouds has begun to decrease~\cite{corneo2021surrounded}. Section 4.1 implies
that for the edge cloud to be beneficial, the edge systems need to be carefully placed such that migrations are minimized. This might require having edge resources that cover larger areas. However, this can complicate the actual deployments as larger coverage implies larger workloads, which in turn requires larger number of servers at the edge. If the cloud latency drops below a certain threshold through deployment of regional data centers, it has the potential (at low to medium utilization levels) to offer overall response times that are “good enough”
for edge applications and lower than what a smaller edge data center can provide (since the reduction
in network latency at an edge site no longer offsets the higher queuing delays in the edge data center). In addition, in the case when migrations are required, the amount of data transferred needs to be minimal in order to not cause severe delays.

Section 4.3 implies that spikes and high variability in the edge workload inter-arrival (and service) times also affects the latency trade-offs between choosing to deploy a workload on the edge or in a backend cloud.
When the workload is bursty in nature, the non-stationarity of the arrival rate  will be high. This implies that edge clouds are less suitable for applications where the request rates  vary considerably, e.g., due to mobility, and especially when these variations occur over a short period of time.

\noindent \textbf{Cost Implications:} 
The deployment and operation of edge-cloud infrastructures introduce several cost implications that must be carefully analyzed to optimize performance while maintaining economic feasibility. Compared to centralized cloud data centers, edge nodes are often deployed in distributed locations, increasing the costs associated with site acquisition, hardware procurement, and network integration. Moreover, managing multiple distributed networks of edge nodes entails higher monitoring and maintenance expenses than a centralized cloud infrastructure. Dynamic workload distribution between edge and cloud resources presents cost trade offs. Over-provisioning edge resources leads to under-utilization and wasted investment, while under-provisioning results in performance bottlenecks and potential SLA violations. The ability to dynamically scale workloads between edge and cloud is essential for cost optimization, requiring intelligent resource management. Application designers should consider these trade offs when designing their edge-cloud usage. While edge computing reduces latency and enhances real-time processing, the high costs of deployment and operation may make it unfavorable for certain use cases. In contrast, relying solely on the cloud may lead to increased latency issues. An optimal cost-performance balance requires strategic workload placement, hybrid processing models, and effective cost-benefit analysis. In addition, the ability of the edge to host multiple applications can significantly reduce the edge costs. To provide higher packing factors, Rack-Scale computers can be used~\cite{novakovic2016case} at the edge. These larger edge resources would reduce the scaling-factors significantly. Smaller and wimpy edge computing resources might be useful for some applications confined within a small geographic area, however, as has been known for years in the modeling research community, resource pooling provides many benefits to distributed processing systems~\cite{tsitsiklis2012power}.

\noindent \textbf{Capacity Planning Implications:}
Our analytic model shows that the ability to multiplex VMs onto a larger number of servers enable a centralized cloud to be more efficient in its packing ability and require a smaller overall capacity than a distributed edge cloud to service the same workload. This property can also be illustrated through a simple example.  Suppose that each server has a unit capacity of 1, and suppose that the workload consists for placement requests for 4 VMs of capacity 0.8, 0.8, 0.2 and 0.2.  In a centralized cloud, these four VMs can be packed onto two servers, each hosting a VM of capacity 0.8 and 0.2.  In a distributed edge cloud with two sites, assuming that the first two requests arrive at the first edge location and the remaining two at the second location, we see that a total of three servers are needed---two at the first edge location, each hosting a VM of size 0.8 and one server at the second location to host two VMs, each of size 0.2.  Thus, the edge cloud needs more servers than the centralized cloud for the same workload and our analytic results provide an expression to quantify the amount of over-provisioning needed in the edge cloud with respect to the centralized cloud.  Our corollary $6.3$ gives a simple $(1+1/q)$ rule to quantify the amount of over-provisioning.
Since $q$ represents the VM packing factor, if we assume that each server can pack a minimum of two VMs (i.e., q=2), 
it follows the edge cloud needs a factor of 1.5 more capacity than a centralized cloud (also illustrated in the above example). A higher packing factor $q$ and more dense packing of VMs implies  a proportionately smaller amount of edge over-provisioning. We note that these results hold for a spatially uniform distribution of VM requests across edge cloud sites.

\section{Related Work}
 Research on modeling clouds and data centers has attracted considerable attention~\cite{khazaei2012performance,bruneo2014stochastic,khazaei2013analysis,gandhi2010optimality,gandhi2013exact}. This has formed the basis for some more recent work to analytically study edge clouds, e.g., to decide where and when services should be migrated in response to user mobility and demand variation~\cite{urgaonkar2015dynamic},  analytical models to compare the performance and utilization between single level and hierarchical designs of the edge clouds~\cite{tong2016hierarchical}, and models to capture the energy consumption trade-offs when offloading the computations or running them locally~\cite{mao2016power}.
 Tsitsiklis and Xu~\cite{tsitsiklis2011power,tsitsiklis2012power} analyze a multi-server model capturing the trade-offs between using centralized and distributed processing using a fluid model approach.  They show that the average queue length in steady state scales as a function of the degree of the fraction of centralized servers $p$ and the traffic intensity, $\lambda$, as $log_{\frac{1}{1-p}}1/(1-\lambda)$ when the traffic intensity approaches 1, which is exponentially smaller than M/M/1 delay scaling.

\section{Conclusion}

In this paper, we analytically study some of the key trade-offs involved in offloading computations to edge clouds versus centralized clouds. Our study focuses on three important aspects that influence performance and cost. First, we analyze how mobility affects edge cloud performance, leading to increased response times. Second, we explore the impact of temporal workload variability and its effect on overall delay. Third, we examine how the VM packing efficiency of edge clouds is reduced compared to centralized clouds due to limited VM multiplexing capabilities at the edge. We believe that our results have significant practical implications for the future design and implementation of resource management algorithms in edge computing by providing a deeper understanding of these trade-offs.

\bibliographystyle{ACM-Reference-Format}
\bibliography{sample-base}

\appendix

\section{Deriving the Squared Coefficient of Variation for $GI/G/1$ with Second Optional Service} \label{appA}

In this section we derive the squared coefficient of variation of the total service time \( S \), denoted by:

\begin{equation}
c_S^2 = \frac{\text{Var}(S)}{(\mathbb{E}[S])^2} \label{eq:c_s_2}
\end{equation}

The service includes a mandatory first phase and an optional second phase, the distribution of \( S \) becomes a mixture.

\[
S =
\begin{cases}
v_1, & \text{with probability } 1 - r \\
v_1 + v_2, & \text{with probability } r
\end{cases}
\]

Where $v_1$ and $v_2$ are independent service times for the mandatory and optional phases, respectively, and $r$ is the probability that a job proceeds to the second phase.

The expected total service time is:

\begin{equation}
\mathbb{E}[S] = \mathbb{E}[v_1] + r \cdot \mathbb{E}[v_2] \label{eq:e_s}
\end{equation}

To model the conditional execution of the second service phase, we introduce a Bernoulli random variable \( B \sim \text{Bernoulli}(r) \), such that:

\[
B =
\begin{cases}
0, & \text{then } S = v_1 \quad \text{(with probability } 1 - r) \\
1, & \text{then } S = v_1 + v_2 \quad \text{(with probability } r)
\end{cases}
\]

We now compute $\text{Var}(S)$ using the law of total variance,

\begin{equation}
\text{Var}(S) = \mathbb{E}[\text{Var}(S \mid B)] + \text{Var}(\mathbb{E}[S \mid B]) \label{eq:var_S}
\end{equation}

Let the variances of the mandatory and optional service phases be $\text{Var}(v_1)$ and $\text{Var}(v_2)$ respectively,

The conditional variances of the total service time \( S \) given the outcome of \( B \) are then:

\[
\text{Var}(S \mid B = 0) = \text{Var}(v_1)
\]
\[
\text{Var}(S \mid B = 1) = \text{Var}(v_1 + v_2) = \text{Var}(v_1) + \text{Var}(v_2)
\]

\noindent
where \( v_1 \) and \( v_2 \) are independent.

Hence, the first term of \ref{eq:var_S} becomes:

\begin{equation}
\mathbb{E}[\text{Var}(S \mid B)] 
= (1 - r)\text{Var}(v_1) + r(\text{Var}(v_1) + \text{Var}(v_2)) 
= \text{Var}(v_1) + r \cdot \text{Var}(v_2) \label{eq:e_var_sb}
\end{equation}

We now compute the second term of \ref{eq:var_S}, $\text{Var}(\mathbb{E}[S \mid B])$:

\[
\mathbb{E}[S \mid B] =
\begin{cases}
\mathbb{E}[v_1], & \text{with probability } 1 - r \\
\mathbb{E}[v_1] + \mathbb{E}[v_2], & \text{with probability } r
\end{cases}
\]

We compute its variance using:

\begin{equation}
\text{Var}(X) = \mathbb{E}[X^2] - (\mathbb{E}[X])^2 \label{eq:var_x}
\end{equation}

Computing the mean:

\begin{equation}
\mathbb{E}[\mathbb{E}[S \mid B]] 
= \mathbb{E}[v_1](1 - r) + (\mathbb{E}[v_1] + \mathbb{E}[v_2]) r 
= \mathbb{E}[v_1] + r \cdot \mathbb{E}[v_2] \label{eq:e_e_sb}
\end{equation}

Computing the expected square:

\begin{equation}
\mathbb{E}[(\mathbb{E}[S \mid B])^2] 
= \left(\mathbb{E}[v_1]\right)^2 (1 - r) + \left(\mathbb{E}[v_1] + \mathbb{E}[v_2]\right)^2 r \label{eq:e_e_sb2}
\end{equation}

\[
= \left(\mathbb{E}[v_1]\right)^2 + 2r\, \mathbb{E}[v_1] \mathbb{E}[v_2] + r\, \left(\mathbb{E}[v_2]\right)^2
\]

Now putting \ref{eq:e_e_sb} and \ref{eq:e_e_sb2} in \ref{eq:var_x}:

\begin{equation}    
\text{Var}(\mathbb{E}[S \mid B]) 
= \left(\mathbb{E}[v_1]\right)^2 + 2r \mathbb{E}[v_1] \mathbb{E}[v_2] + r \left(\mathbb{E}[v_2]\right)^2 
- \left(\mathbb{E}[v_1] + r \mathbb{E}[v_2]\right)^2 
= r(1 - r) \left(\mathbb{E}[v_2]\right)^2
\label{eq:var_e_sb}
\end{equation}

Putting \ref{eq:e_var_sb} and \ref{eq:var_e_sb} into \ref{eq:var_S}:

\begin{equation}
\text{Var}(S) = \text{Var}(v_1) + r \cdot \text{Var}(v_2) + r(1 - r)(\mathbb{E}[v_2])^2 \label{eq:var_s_final}
\end{equation}

Putting \ref{eq:var_s_final} and \ref{eq:e_s} in \ref{eq:c_s_2}:

\[
c_S^2 
= \frac{
\text{Var}(v_1) + r \cdot \text{Var}(v_2) + r(1 - r)(\mathbb{E}[v_2])^2
}{
\left( \mathbb{E}[v_1] + r \cdot \mathbb{E}[v_2] \right)^2
}
\]

\section{Deriving latency bound for $GI/G/k$ system}
\label{AppGG1}
To model the expected waiting time $E(w)$ for a $GI/G/1$ queue, we use the Allen-Cunneen approximation~\cite{bolch2006queueing,KELLYBOOTLE1990247}, a result that has been shown to be accurate compared to exact values for heavy load regimes~\cite{whitt1993approximations}, and has found many practical applications~\cite{ahmad2010joint,hong2011dynamic}. Allen-Cunneen approximation provides a correction term over $M/M/1$ as,

\begin{equation}
\mathbb{E}[w]_{GI/G/1}\approx 
\mathbb{E}(w)_{M/M/1} \cdot
\frac{c^2_A+c^2_S}{2} \Rightarrow\frac{\rho}{\mu(1-\rho)}.\frac{c^2_A+c^2_S}{2},\label{eq:Allen1} 
\end{equation}

where, $c_A^2$ and $c_S^2$ are the squared Coefficients of Variation (CoV) of the inter-arrival time and the service times respectively. 

From \ref{MMk Latency bound}, we already have the  $\mathbb{E}(w)_{M/M/1}$ with second optional service to incorporate mobility. To extend it with the Allen-Cunneen approximation \eqref{eq:Allen1}, we must incorporate the effect of both $\mu_1$ and $\mu_2$ together in the $c_S^2$.  We derive the squared coefficient of variation of the total service time with the help of Law of Total Variance in Appendix \ref{appA}.

\begin{equation}
c_S^2 
= \frac{\text{Var}(S)}{(\mathbb{E}[S])^2} = \frac{
\text{Var}(v_1) + r \cdot \text{Var}(v_2) + r(1 - r)(\mathbb{E}[v_2])^2
}{
\left( \mathbb{E}[v_1] + r \cdot \mathbb{E}[v_2] \right)^2 \label{eq:cs2} 
},
\end{equation}
where $\mathbb{E}[v_1]$, $\mathbb{E}[v_2]$, $\text{Var}(v_1)$, and  $\text{Var}(v_2)$ are the mean and variance of two service times, $v_1$ and $v_2$, respectively. 


With $c_S^2$ defined as \eqref{eq:cs2},
\begin{equation}
\mathbb{E}[w]_{GI/G/1} \approx \frac{\lambda \left( \frac{1}{\mu_1^2} + \frac{r}{\mu_2^2} + \frac{r}{\mu_1 \mu_2} \right)}
{1 - \frac{\lambda}{\mu_1} - \frac{r \lambda}{\mu_2}}\cdot  
\frac{c_A^2 + c_S^2}{2} 
\label{eq:gg1_optional2}
\end{equation}

If there is no migration, i.e.\ $r = 0$ or migration does not take any time i.e. $\mu_2 \rightarrow \infty$, then equation \ref{eq:gg1_optional2} becomes
\( \frac{\rho}{\mu(1-\rho)}.\frac{c^2_A+c^2_S}{2}
\), which is the Allen-Cunneen approximation.

Similar to section \ref{MMk Latency bound}, we add the impact of the destination server utilization, and migration service time.
\begin{equation}\mathbb{E}[w_{dest}]_{(GI/G/1)} \approx 
\frac{r\lambda}{\mu_1 (\mu_1- r\lambda)}.\frac{c^2_{A_{dest}}+c^2_{S_{dest}}}{2}
\end{equation}

The total expected waiting time at the edge is given by:
\begin{equation}
\mathbb{E}[w_{edge}]_{GI/G/1} \approx \frac{\lambda \left( \frac{1}{\mu_1^2} + \frac{r}{\mu_2^2} + \frac{r}{\mu_1 \mu_2} \right)}
{1 - \frac{\lambda}{\mu_1} - \frac{r \lambda}{\mu_2}}\cdot  
\frac{c_A^2 + c_S^2}{2} + \frac{r\lambda}{\mu_1 (\mu_1- r\lambda)}.\frac{c^2_{A_{dest}}+c^2_{S_{dest}}}{2}
\label{eq:gg1_optional221}
\end{equation}

Assuming that the inter-arrival and service time variances are similar on both source and destination sites, 
\begin{equation}
\mathbb{E}[w_{edge}]_{GI/G/1} \approx \left[\frac{\lambda \left( \frac{1}{\mu_1^2} + \frac{r}{\mu_2^2} + \frac{r}{\mu_1 \mu_2} \right)}
{1 - \frac{\lambda}{\mu_1} - \frac{r \lambda}{\mu_2}}  
 + \frac{r\lambda}{\mu_1 (\mu_1- r\lambda)}\right]\cdot\frac{c_{A_{edge}}^2 + c_{S_{edge}}^2}{2}
\label{eq:gg1_optional22}
\end{equation}

For G/G/k systems, the approximation for the expected waiting time  is given by ~\cite{bolch2006queueing, whitt1993approximations},
\begin{equation}
    \mathbb{E}[w]_{G1/G/k} \approx \frac{P_w}{\mu(1-\rho)}.\frac{c^2_A+c^2_S}{2k}, \label{eq:Allen2}  
\end{equation}
where $P_w$ is the steady state probability that an arriving request has to wait in the queue for a server to become available. Bolch et al.~\cite{bolch2006queueing} have shown that $P_w$ can be approximated closely as,
\begin{equation}
    P_w \approx \begin{cases}
    \frac{\rho^k +\rho}{2} , & \text{if}\  \rho>0.7. \label{eq:Bosch1} \\
    \rho^\frac{k+1}{2},  & \text{if}\  \rho<0.7    
    \end{cases}
\end{equation}

Since the Allen-Cunneen approximation is more accurate in heavy traffic regimes, and since higher utilization, specially at the edge, has a larger impact on the performance, we will only consider the heavy traffic regime from Equation~\ref{eq:Bosch1}, i.e., $\rho>0.7$. 

\begin{equation}
    \mathbb{E}[w_{cloud}]_{GI/G/k} \approx \frac{\rho_{cloud}^k + \rho_{cloud}}{2\mu_{cloud}(1-\rho_{cloud})}.\frac{c^2_{A_{cloud}}+c^2_{S_{cloud}}}{2k}, \label{eq:cloudggk}  
\end{equation}

Putting $w_{edge}$, $s_{\text{migration}}$ and $w_{cloud}$ in \ref{eq:deltan2}, we get,



\begin{equation}
\Delta t > [\underbrace{
\frac{\lambda \left( \frac{1}{\mu_1^2} + \frac{r}{\mu_2^2} + \frac{r}{\mu_1 \mu_2} \right)}{1 - \frac{\lambda}{\mu_1} - \frac{r \lambda}{\mu_2}}
+ \frac{r\lambda}{\mu_1 (\mu_1- r\lambda)}] \cdot  
\frac{c_A^2 + c_S^2}{2}}_{w_{\text{edge}}}
 + \underbrace{\frac{r}{\mu_2}}_{s_{\text{migration}}} 
- \underbrace{
\frac{\rho_{\text{cloud}}^k + \rho_{\text{cloud}}}{2\mu_{\text{cloud}}(1-\rho_{\text{cloud}})} \cdot \frac{c^2_{A_{\text{cloud}}}+c^2_{S_{\text{cloud}}}}{2k}
}_{w_{\text{cloud}}}, \label{deltat_gg1_final}
\end{equation}

which completes the proof.
    
\section{Deriving excess waiting time due to sinusoidal spike}\label{appspike}
Eick \cite{eick1993mt} and Whitt \cite{whitt2014steady}, studied \( M_t/M/\infty \) queues with sinusoidal arrival rates, and showed that the number of busy servers follows a Poisson distribution. As a result, the system can be characterized by its mean occupancy \( m(t) \), which varies periodically with the arrival rate \( \lambda(t) \) as follows,

\begin{equation}
m(t) \;=\; \frac{\bar{\lambda}}{\mu_{\mathrm{eff}}}
\;\Biggl[\,
1
\;+\;
\frac{A}{1 + \bigl(\tfrac{\gamma}{\mu_{\mathrm{eff}}}\bigr)^{2}}
\Bigl(
\sin\bigl(\gamma t)
\;-\;
\frac{\gamma}{\mu_{\mathrm{eff}}}\,\cos\bigl(\gamma t)
\Bigr)
\Biggr].  \label{eq:mt}  
\end{equation}

The system’s congestion response lags behind the time‐varying demand. Therefore,  \cite{eick1993mt} show that $m(t)$ has a time lag of,
\begin{equation}
\quad lag(\gamma)\;=\;\frac{\cot^{-1}(1/\gamma)}{\gamma},    
\end{equation}

Although equation \eqref{eq:mt} is derived for infinite-server queues, where no waiting occurs, the authors in \cite{eick1993mt} note that it also serves as a good approximation for \( M_t/M/k \) systems. Our simulations show that it remains a reasonable approximation even for \( M_t/M/1 \) queues., provided that \( \lambda(t) < \mu \) at all times, which is the stability condition. 

Since \(m(t)\) denote the expected number of busy servers in an \(M_t/M/\infty \) system at time \(t\), if we have an \(M_t/M/k\) (or \(M_t/M/1\)) system, then on average \(m(t)\) of its \(k\) servers would be busy, so the instantaneous utilization can be approximated by,
\[
\rho_{\mathrm{eff}}(t)\;\approx\;\frac{m(t)}{k}, \quad \text{for } k=1,\quad \frac{m(t)}{1} = m(t)
\]

\[
\text{Let:}\qquad
u=\bar\rho=\frac{\bar\lambda}{\mu},\quad
\beta=\frac{\gamma}{\mu},\quad
C=\frac{A}{1+\beta^{2}},\quad
f(t)=\sin(\gamma t)-\beta\cos(\gamma t),
\]
then,
\begin{equation}
    m(t)=u\,[\,1+C\,f(t)\,].
\end{equation}

For an $M/M/1$ queue, the waiting time is given by \( w_q(\rho) = \frac{\rho}{\mu(1-\rho)} \). Assuming that \( m(t) \) approximates instantaneous utilization \( \rho(t) \), we write:
\begin{equation}
w(t)=\frac{m(t)}{\mu_{eff}\,[1-m(t)]}.\label{eq:wedgesin}    
\end{equation}
Let \( \delta(t) \) be a small perturbation added to the constant \( u \):
\[
m(t)=u+\delta(t),
\qquad
\delta(t):=uC\,f(t).
\]

If we nudge the input of a smooth function (such as \( w_q(\rho) \)) from \( u \) to \( u+\delta(t) \),
the output changes only slightly. We apply a Taylor expansion:

\[
w(u+\delta(t))\;\approx\;
\underbrace{w(u)}_{\text{value at }u}
+\underbrace{w'(u)}_{\text{slope at }u}\,\delta(t)
+\tfrac12\underbrace{w''(u)}_{\text{curvature at }u}\,\delta(t)^{2}.
\]

Here, the slope is the first derivative (rate of change), and the curvature is the second derivative. 
Higher powers \(\delta(t)^{3}, \delta(t)^{4}, \dots\) are increasingly
small and will be grouped into the symbol \(O(A^{3})\).

The required derivatives of \( w_q(\rho) = \frac{\rho}{\mu(1-\rho)} \) are:
\[
w_q'(u)=\frac{1}{\mu(1-u)^{2}},
\qquad
w_q''(u)=\frac{2}{\mu(1-u)^{3}}.
\]

Substituting into the Taylor expansion:

\[
w(t)=w\bigl(u+\delta (t)\bigr)
\approx
\frac{u}{\mu(1-u)}
+\frac{uA\,f(t)}{\mu(1-u)^{2}}
+\frac{u^{2}A^{2}f(t)^{2}}{\mu(1-u)^{3}}
+O(A^{3}).
\]

Because the average of a sinusoid over one complete period \(T\) is:
\[
\langle f(t) \rangle_T = 0, \qquad
\langle f(t)^{2} \rangle_T = \tfrac12(1+\beta^{2}),
\]
the linear term disappears, and we obtain:

\[
\mathbb{E}[w(t)]
=\frac{u}{\mu(1-u)}
+\frac{u^{2}A^{2}}{2\,\mu(1-u)^{3}(1+\beta^{2})}
+O(A^{3}).
\]

The stationary queue delay (\(A=0\)) is:
\[
\mathbb{E}[w]_{\text{stat}}=\frac{u}{\mu(1-u)} .
\]

Hence, the \emph{excess} mean waiting time is:
\begin{equation}
\Delta w_q(A)
   =\frac{u^{2}A^{2}}
          {2\,\mu(1-u)^{3}(1+\beta^{2})}
   +O(A^{3}). 
\end{equation}

\begin{corollary}
If the arrival are sinusoidal, for queue to remain stable, the average utilization $\bar\rho$ must remain below 0.5. \label{sinutilization}
\end{corollary} 
\begin{proof}  
When the sinusoidal amplitude approaches its maximum value, \(A \to 1\), the
arrival rate is,
\[
\lambda_{\max} \;=\; \bar{\lambda}\bigl(1 + A\bigr)
\;=\;
2\,\bar{\lambda}.
\]

System stability requires that the instantaneous arrival rate never exceed the
effective service capacity, hence
\[
2\bar{\lambda} < \mu_{\mathrm{eff}}.
\]

Dividing both sides by \(\mu_{\mathrm{eff}}\),
\[
\bar\rho = \frac{\bar{\lambda}}{\mu_{\mathrm{eff}}}<\frac{1}{2}\]
\end{proof} 

\begin{corollary}   
If the sinusoidal arrivals at edge sites have a uniform diversity in phase shifts $\phi_i$, as $k \rightarrow \infty $ cloud workload becomes stationary, consequently offering lower waiting times. \label{sinsumstatcloud}
\end{corollary}
\begin{proof}
 If edge sites have a diversity in phase shifts $\phi_i$
such that,
\[ \lambda_{edge(i)}(t) = \bar\lambda_{i}(1 + A_{i} \sin(\gamma_i t + \phi_i))),\] 


If, $\phi_i$ are phases drawn from a probability distribution $f(\phi)$, If the phases \( \phi_i \) are independently and identically distributed (i.i.d.) with a uniform distribution over \([0,2\pi]\),
\[
f(\phi) = \frac{1}{2\pi}, \quad \phi \in [0,2\pi],
\]
The sinusoidal terms \( \sin(\gamma t + \phi_i) \) are symmetrically distributed around zero, leading to cancellations in the aggregate sum. 
The expected value of \( \sin(\gamma t + \phi_i) \) is zero:
\[
E[\sin(\gamma t + \phi_i)] = 0,
\]
The variance of the sum scales as \( 1/\sqrt{k} \).
As $k \rightarrow \infty$, $ A_{cloud} \rightarrow 0 $,  making the system more and more stationary and thus by Ross's conjecture reducing the queuing delays.   
\end{proof}
\section{Deriving Fluid Approximation for Net-Input Process with Sinusoidal Arrival Rate} \label{appb}

The arrival rate is given by the sinusoidal form:
\[
\lambda(t) = \bar{\lambda} \left(1 + A \sin(\gamma t)\right),
\]
where \( \bar{\lambda} \) is the average arrival rate, \( A \in [0, 1] \) is the amplitude, and \( \gamma = \frac{2\pi}{T} \) is the frequency over a cycle of length \( T \).

The effective service rate accounting for both service phases is:
\[
\mu_{\text{eff}} = \left( \frac{1}{\mu_1} + \frac{r}{\mu_2} \right)^{-1} = \frac{\mu_1 \mu_2}{\mu_2 + r\mu_1}.
\]

We are interested in estimating the fluid backlog \( Q(t_1, t_2) \) accumulated over an interval \( [t_1, t_2] \), where \( \lambda(t) > \mu_{\text{eff}} \). The net-input during this interval is:
\begin{equation}
V(t_1, t_2) = \int_{t_1}^{t_2}\left[\lambda(t) - \mu_{\text{eff}} \right] \, dt. \label{eq:net_input}
\end{equation}

Since queue lengths cannot be negative, we define the backlog as the \emph{positive part} of the net input:
\[
Q(t_1, t_2) = [V(t_1, t_2)]_+, \qquad \text{where } [x]_+ = \max(x, 0).
\]

\subsection*{Evaluating the Net Input}

Substituting the sinusoidal form of \( \lambda(t) \) into the integrand:
\[
V(t_1, t_2) = \int_{t_1}^{t_2} \left[ \bar{\lambda} - \mu_{\text{eff}} + \bar{\lambda} A \sin(\gamma t) \right] dt.
\]

\[
 \hspace{65pt}  = (\bar{\lambda} - \mu_{\text{eff}})(t_2 - t_1) + \bar{\lambda} A \int_{t_1}^{t_2} \sin(\gamma t) \, dt.
\]

Using the identity \( \int \sin(ax) \, dx = -\frac{1}{a} \cos(ax) + C \), we evaluate:
\[
\int_{t_1}^{t_2} \sin(\gamma t) \, dt = -\frac{1}{\gamma} \left[ \cos(\gamma t_2) - \cos(\gamma t_1) \right].
\]

Substitute into the expression for \( V(t_1, t_2) \):
\begin{equation}
V(t_1, t_2) = (\bar{\lambda} - \mu_{\text{eff}})(t_2 - t_1) + \frac{\bar{\lambda} A}{\gamma} \left[ \cos(\gamma t_1) - \cos(\gamma t_2) \right]. \label{eq_vt}    
\end{equation}

Therefore, the fluid backlog is:
\[
Q(t_1, t_2) = \left[
(\bar{\lambda} - \mu_{\text{eff}})(t_2 - t_1) + \frac{\bar{\lambda} A}{\gamma} \left( \cos(\gamma t_1) - \cos(\gamma t_2) \right)
\right]_+.
\]

Finally, the fluid approximation for the average waiting time over this interval is:
\[
W_f(t_1, t_2) = \frac{Q(t_1, t_2)}{\mu_{\text{eff}}}.
\]

\subsection*{Identifying the Overload Interval \texorpdfstring{\( [t_1, t_2] \)}{[t1, t2]} for Sinusoidal Arrivals}

The time interval during which the arrival rate \( \lambda(t) \) exceeds the effective service rate \( \mu_{\text{eff}} \) is denoted by \( [t_1, t_2] \). In this overload interval, the system accumulates backlog because it cannot serve jobs as quickly as they arrive. The endpoints \( t_1 \) and \( t_2 \) are defined by the intersection of \( \lambda(t) \) with \( \mu_{\text{eff}} \).

\[
\lambda(t) = \bar{\lambda} \left(1 + A \sin(\gamma t)\right) = \mu_{\text{eff}},
\]

we solve for \( t \):

\[
1 + A \sin(\gamma t) = \frac{\mu_{\text{eff}}}{\bar{\lambda}} \]
\[
\sin(\gamma t) = \frac{\mu_{\text{eff}}}{\bar{\lambda}} - 1.
\]

\[
\gamma t = \arcsin( \frac{\mu_{\text{eff}}}{\bar{\lambda}} - 1).
\]

Let:

\[
\theta := \arcsin\left(\frac{\mu_{\text{eff}}}{\bar{\lambda}} - 1\right).
\]

Then the two crossing points are:
\[
t_1 = \frac{\theta}{\gamma}, \qquad
t_2 = \frac{\pi - \theta}{\gamma}.
\]

\subsection*{Substituting \( t_1 \) and \( t_2 \) into \( V(t_1, t_2) \)}





Compute \( t_2 - t_1 \):
\begin{equation} 
t_2 - t_1 = \frac{\pi - \theta}{\gamma} - \frac{\theta}{\gamma} = \frac{\pi - 2\theta}{\gamma}   \label{eqt2t1}
\end{equation}
Compute \( \cos(\gamma t_1) - \cos(\gamma t_2) \)

We substitute:
\[
\gamma t_1 = \theta, \quad \gamma t_2 = \pi - \theta
\]

Using the identity \( \cos(\pi - \theta) = -\cos(\theta) \), we get:
\begin{equation}
\cos(\gamma t_1) - \cos(\gamma t_2) = \cos(\theta) - (-\cos(\theta)) = 2\cos(\theta) \label{cosgamma}
\end{equation}

Plug \eqref{eqt2t1} and \eqref{cosgamma}  into the backlog expression \eqref{eq_vt}:
\begin{equation}
V(t_1,t_2) = \frac{1}{\gamma} \left( (\bar{\lambda} - \mu_{\text{eff}})(\pi - 2\theta)
+ 2\bar{\lambda} A \cos(\theta) \right),   
\quad \text{where } \theta = \arcsin\left( \frac{\mu_{\text{eff}}}{\bar{\lambda}} - 1 \right)
\end{equation}

\end{document}